\documentclass[11pt,letterpaper]{article} 
\usepackage[letterpaper,margin=1in]{geometry}

\usepackage[UKenglish]{babel}
\usepackage{fullpage}

\usepackage{amsmath,amssymb,amsthm}
\usepackage{xcolor}
\usepackage[colorlinks=true, allcolors=blue]{hyperref}
\usepackage[capitalise,noabbrev]{cleveref}
\usepackage{thmtools} 
\usepackage{thm-restate}
\usepackage{enumerate}

\newtheorem{theorem}{Theorem}[section]

\newtheorem{lemma}[theorem]{Lemma}
\newtheorem{definition}[theorem]{Definition}

\newtheorem{claim}[theorem]{Claim}

\newtheorem{observation}[theorem]{Observation}
\newtheorem*{theorem*}{Theorem}

\newcommand{\eps}{\varepsilon}
\renewcommand{\phi}{\varphi}

\DeclareMathOperator*{\polylog}{polylog}

\DeclareMathOperator{\OPT}{OPT}
\DeclareMathOperator*{\area}{area}
\DeclareMathOperator*{\width}{width}
\DeclareMathOperator*{\height}{height}
\DeclareMathOperator*{\base}{base}
\DeclareMathOperator*{\shadow}{shadow}

\newcommand{\OnlineSorting}{\textsc{Online-Sorting}}
\newcommand{\OnlineTSPordering}{\textsc{On\-line-TSP-or\-de\-ring}}
\newcommand{\OnlineTSPscheduling}{\textsc{On\-line-TSP-Sche\-du\-ling}}

\usepackage{todonotes}

\title{Online Geometric Packing through Online TSP Scheduling} 
\date{}

\author{Anders Aamand \and Mikkel Abrahamsen \and Simon Bartlmae \and Arindam Khan \and Linda Kleist\and Csaba D. T\'{o}th}

\begin{document}

\maketitle

\begin{abstract}
We consider the problem of online packing of convex polygons into a strip by translations. While online algorithms with a constant competitive ratio have been known for rectangles for decades [Baker and Schwarz, SICOMP 1983],
the current best algorithm for convex polygons has competitive ratio $O(n^{\log_2 3-1}\log n) = O(n^{0.59})$, where $n$ is the number of polygons.
This algorithm was described by Aamand, Abrahamsen, Beretta, and Kleist [SODA 2023], who also proved a lower bound of  $\Omega(\sqrt{\log n/\log\log n})$ on the competitive ratio of any algorithm.
Their lower bound is obtained via a reduction from \emph{online sorting}, a problem introduced in the same paper, for which they established a lower bound on the competitive ratio.

We introduce a new, natural online problem that we call \emph{online TSP scheduling}.
Here, points $x_1,\ldots,x_n$ arrive online from a metric space $(M,d)$, and upon arrival each $x_i$ must be assigned a visit time $p_i\in[0,\infty)$ satisfying $|p_i-p_j|\ge d(x_i,x_j)$ for all $j<i$.
The cost of the schedule is $\max_i p_i$.
We present an $O(\log^2 n)$-competitive algorithm for online TSP scheduling, and show how this implies an $O(\log^2 n)$-competitive algorithm for online translational strip packing of convex polygons. We also prove that the same competitive ratio is achievable for other translational packing problems, including online packing of $d$-dimensional unit hyperdisks in $\mathbb R^{d+1}$, whose offline version was studied by Alt, Cabello, Cheong, Park, and Seiferth [Comp. Geom. 2026].

Our algorithm for online TSP scheduling builds on a recent breakthrough for online sorting by Azar, Panigrahi, and Vardi [SODA 2026].
We thus show that the connection between packing and online sorting can be used not only for lower bounds, but also for algorithms.

In a related problem, \emph{online TSP ordering}, we have to place the points $x_1,\ldots,x_n$ in empty cells in an array of size $(1+\eps)n$, for constant $\eps>0$, and the cost of a solution is the total length of the path defined by the points in left-to-right order in the array.
Our methods lead to an algorithm with competitive ratio $O((\log^3 n)/\eps)$ for online TSP ordering, making progress on a question by Bertram [ESA 2025].

\end{abstract}

\thispagestyle{empty}

\newpage
\thispagestyle{empty}
\tableofcontents

\newpage
\setcounter{page}{1}

\section{Introduction}
The problem of efficiently packing objects within a given container is a fundamental question in computational geometry, with wide-ranging applications in logistics, manufacturing, materials science, and robotics. 
While many real-world applications involve objects with high complexity, 
the vast majority of theoretical investigations focus on axis-aligned rectangular objects; see surveys~\cite{CHRISTENSEN201763,epstein2018multidimensional} and more recent results~\cite{AbrahamsenBeretta20,balogh2019lower, BansalK14, epstein2019lower, GalvezGIHKW21, Kar0R25, KarKW26}.
In some practical settings, objects must be placed with a fixed orientation, meaning that only translations are allowed (i.e., reflections and rotations are prohibited).
For example, in garment production, complex pieces are cut from a roll of fabric and later sewn together.
The orientation of each piece is constrained by the fabric grain, and therefore cannot be rotated arbitrarily.
Such questions from clothing production motivated Milenkovic~\cite{Milenkovic96,DBLP:journals/algorithmica/Milenkovic97,MILENKOVIC19993} and Milenkovic and Daniels~\cite{doi:10.1111/j.1475-3995.1999.tb00171.x} to study packing problems involving irregular shapes such as arbitrary convex or simple polygons.
They gave exponential-time offline algorithms finding optimal packings for various versions of the problem.
Since then, theoretical research on irregular packing problems has been mostly absent, but in recent years, there has been a growing effort in studying the theory of packings beyond rectangular objects to encompass more complex shapes, such as arbitrary simple polygons~\cite{DBLP:journals/theoretics/AbrahamsenMS24}, convex polygons~\cite{aamand2023online,alt_convexOffline_JoCG,alt_convexOffline_corr,DBLP:conf/esa/KurpiszS23,DBLP:conf/icalp/MerinoW20}, orthogonal polygons~\cite{OrthoOnlineSoCG}, disks and hyperspheres \cite{AcharyaBG0MW24, MiyazawaPSSW16}.

\paragraph*{The problem.} In this paper, we study online translational packing problems where the pieces to be packed are general convex polygons. Let us focus for now on one of these packing variants, called~\emph{strip packing}. The strip is the region $S := [0,\infty) \times [0,1]$. The input consists of a finite sequence $
\mathcal{I} = (P_1,P_2,\dots,P_n)$, where each $P_i$ is a convex polygon in the plane. The input is revealed sequentially and the task is to choose an irrevocable placement of polygon~$P_i$ in the strip $S$ before $P_{i+1}$ is revealed.
Each polygon can only be translated but neither rotated nor reflected.
The placements of the polygons must constitute a \textit{packing}, where the interiors of any two polygons are disjoint.
The objective is to minimize the cost of a packing defined as the maximum $x$-coordinate over all occupied points in $S$. \cref{fig:intro} illustrates a packing and its cost. 

\begin{figure}[htb]
    \centering
    \includegraphics[page=1]{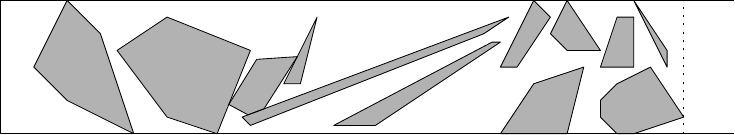}
    \caption{A strip packing of convex polygons}
    \label{fig:intro}
\end{figure}

Aamand, Abrahamsen, Beretta, and Kleist~\cite{aamand2023online}
showed that for online strip packing of convex polygons (even for parallelograms with diameter bounded by an arbitrary small constant), no online algorithm can achieve a constant competitive ratio. More precisely, they proved that the competitive ratio of any online algorithm is $\Omega(\sqrt{\log n/\log\log n})$ and that there exists an online algorithm with competitive ratio of $O(n^{\log_2 3-1}\log n)=O(n^{0.59})$, where $n$ is the number of polygons. They obtained the lower bound by reducing a natural online sorting problem to strip packing convex polygons and providing lower bounds for the sorting problem.
In online sorting, we are given an empty array and receive $n$ numbers from $[0,1]$ one after the other.
Each number has to be placed irrevocably in an empty cell in the array.
We consider the non-empty cells from left to right, and the goal is to minimize the sum of the pairwise differences of consecutive numbers in this order.
Placing the points in sorted order is clearly optimal, hence the name of the problem.
The online sorting problem has sparked several follow-up investigations \cite{DBLP:conf/esa/AbrahamsenB0K024,azar2026nearly, DBLP:conf/esa/Bertram25, DBLP:conf/soda/000126b, DBLP:conf/stacs/KalavasPT26,DBLP:conf/waoa/NirjhorW25}.

Online strip packing of axis-parallel rectangles was first studied by Baker and Schwarz~\cite{baker1983shelf}, followed by several other papers~\cite{DBLP:journals/acta/BrownBK82,Csirik1997shelf,han2007strip,DBLP:journals/iandc/HanIYZ16,HarrenK15,DBLP:journals/tcs/HurinkP11,KernP13,YeOnlineStrip}.
When rotations are allowed, packing problems often admit algorithms with better competitive ratios.
For instance, Lassak and Zhang~\cite{DBLP:journals/dcg/LassakZ91} described an $O_d(1)$-competitive online algorithm for packing $d$-dimensional convex objects of bounded diameter into a minimum-volume cube using rotations; such an algorithm does not even exist for $d=2$ when only translations are allowed by results from~\cite{aamand2023online}.

\subsection{Our contributions} 
In this paper, we show that ideas from online sorting (and the more general online TSP scheduling problem to be introduced here) can be used to devise algorithms for more complex online translational packing problems. 
In a sense, this reverses the perspective of the lower-bound construction of Aamand, Abrahamsen, Beretta, and Kleist~\cite{aamand2023online}: whereas they used online sorting to prove lower bounds for online packing problems, we use algorithmic progress on online-sorting-like problems to obtain improved packing algorithms.

\paragraph*{Two-dimensional online translational packing problems.}
Our first result considers two-dimen\-sion\-al online translational packing.
We develop a framework that gives results for several well-known packing variants.
In \textit{bin packing}, the arriving pieces must be packed into unit square containers and the goal is to minimize the number of used containers.
In \textit{perimeter packing} and \textit{area packing}, the pieces can be packed anywhere in the plane, and the goal is to minimize respectively the perimeter and the area of their axis-aligned bounding box.

\begin{restatable}{theorem}{main}\label{thm:main}
There exist algorithms with the following competitive ratios for translational online packing of $n$ convex polygons:
\begin{enumerate}[(i)]\itemsep 0pt
\item \label{item:0} an $O(\log^2 n)$-competitive algorithm for strip packing.
\item \label{item:i} an $O({(\log^2 n)/\delta})$-competitive algorithm for bin packing, as long as each arriving piece has horizontal span at most $1-\delta$. 
\item \label{item:ii} an $O(\log n)$-competitive algorithm for perimeter packing.
\item \label{item:iii} an $O(\sqrt{n}\cdot \log^2 n)$-competitive algorithm for area packing.
\end{enumerate}
The algorithms are deterministic and do not need to know $n$ in advance.
\end{restatable}
Our result on strip packing is an exponential improvement over the state of the art~\cite{aamand2023online}.
No algorithms seem to have been described for the other variants when the pieces are convex polygons.
Online bin packing of rectangles was first studied by Coppersmith and Raghavan~\cite{coppersmith1989multidimensional} with numerous later works improving upper and lower bounds~\cite{balogh2019lower, csirik1993line,epstein2019lower,DBLP:journals/siamcomp/EpsteinS05,DBLP:journals/talg/HanCTZZ11,DBLP:journals/algorithmica/SeidenS03,DBLP:journals/sigact/Stee15}, and perimeter and area packing have also been studied for squares and rectangles~\cite{AbrahamsenBeretta20,fekete2017online}.
There are known lower bounds of $\Omega(\sqrt{\log n/\log \log n})$ for strip and bin packing and $\Omega(\sqrt[4]{\log n/\log \log n})$ for perimeter packing~\cite{aamand2023online}.
Concerning area packing, there is a lower bound of $\Omega(\sqrt{n})$ already when the pieces are axis-parallel rectangles~\cite{AbrahamsenBeretta20}.
Closing the gaps between the upper and lower bounds remains an open problem.

\paragraph*{Online TSP scheduling.}
We obtain the above results on online translational packing problems through an understanding of a new variant of online TSP which is likely of independent interest.
\begin{definition}[Online TSP Scheduling]
Let $(M,d)$ be a metric space. Points $x_1,\ldots,x_n\in M$ arrive online,
and upon arrival, each $x_i$ must be assigned a visit time
$p_i\in[0,\infty)$ satisfying $|p_i-p_j|\ge d(x_i,x_j)$ for all $j<i$. The cost of the schedule is $\max_i p_i$.
\end{definition}

We denote this problem by \OnlineTSPscheduling$[n,M,d]$ and its
offline optimum by $\OPT$. Ordering the points by their visit times $p_i$ shows
that every feasible schedule induces a traveling-salesperson path of length at
most $\max_{i}p_i$. Conversely, any traveling-salesperson path can be scheduled with
cost equal to this length. Hence, $\OPT$ equals the length of a shortest
traveling-salesperson path through the input points. 

The \OnlineTSPscheduling~problem has a natural interpretation.
Imagine a person getting requests to visit different locations in an online fashion.
Upon receiving a request to visit $x_i\in M$, the time $p_i\in[0,\infty)$ at which $x_i$ is visited has to be scheduled immediately and irrevocably.
The underlying metric measures the time it takes to travel between different locations, and the constraint $|p_i-p_j|\ge d(x_i,x_j)$ ensures that there is always sufficient travel time. 

The problem can also be seen as an online variant of \emph{scheduling with sequence-dependent setup times (SSDST)}, which has a rich literature in operations research; see the surveys~\cite{Allahverdi1999, allahverdi2008survey} and subsequent work~\cite{LeibHellerKuehn2025,LinYing2022}.
In SSDST, jobs must be scheduled on a machine.
Job $i$ has processing time~$t_i$, and for every pair of jobs $(i,j)$, there is a setup time $\rho(i,j)$.
The algorithm has to assign a start time $s_i$ to job~$i$.
The resulting schedule must be feasible: jobs cannot overlap, and if job $j$ is scheduled after job $i$, then there has to be at least $\rho(i,j)$ time between the completion of $i$ and the start of $j$.
The objective is to minimize the makespan.
In many applications, the time needed to switch from one job to another depends on the two jobs involved, for example due to cleaning, color changes, temperature changes, tooling, or calibration, which is represented by the $\rho$-values.
Assuming $\rho$ is a metric, we can consider \OnlineTSPscheduling\ with the following metric: $\delta(i,j):= \frac{t_i}{2}+\rho(i,j)+\frac{t_j}{2}$ for $i \neq j$ and $\delta(i,i)=0$. Again, the constraint $|p_i-p_j|\ge d(x_i,x_j)$ ensures that sufficient setup time is left between jobs.

Our main result on online TSP scheduling is as follows.
\begin{restatable}{theorem}{mainTSPscheduling}\label{thm:TSP:scheduling}
There exists an algorithm for
\OnlineTSPscheduling$[n,M,d]$ with competitive ratio
$O(\log^2 n)$.
The algorithm does not need to know $n$ in advance.
\end{restatable}
Our proof uses \textit{elementary trees}, recently introduced by Azar, Panigrahi, and Vardi~\cite{azar2026nearly} to study the online sorting problem of~\cite{aamand2023online}. However, to handle the case of general metric spaces we need several new ideas.
We point out that further improving the competitive ratio in \cref{thm:TSP:scheduling} will improve all the competitive ratios of the packing algorithms in \cref{thm:main}.

\paragraph*{Online TSP ordering.}
Abrahamsen, Bercea, Beretta, Klausen and Kozma~\cite{DBLP:conf/esa/AbrahamsenB0K024} introduced another natural generalization of the online sorting problem, that we call\footnote{The authors of~\cite{DBLP:conf/esa/AbrahamsenB0K024} dubbed this problem~\emph{Online TSP}. However, this term usually refers to the online traveling salesperson problem of Ausiello, Feuerstein, Leonardi, Stougie, and Talamo~\cite{ausiello2001algorithms} which is of a very different nature. To distinguish between the problems, we therefore use the suffix `\textsc{-ordering}' to indicate that the order of the points in the array determines the cost.} \OnlineTSPordering{} 
in this paper.

\begin{definition}[Online TSP Ordering]
Let $(M,d)$ be a metric space, $n$ a natural number, and $\eps\geq 0$. Points $x_1,\ldots,x_n\in M$ arrive online,
and upon arrival each $x_i$ must be assigned to an empty cell of an array $A$ consisting of $(1+\eps)n$ cells. The cost  of the ordering is $\sum_{i=1}^{n-1}d(z_i,z_{i+1})$ where $z_1,\dots,z_n$ are the points in non-empty cells of $A$ in left-to-right order.
\end{definition}
To see how \OnlineTSPscheduling~differs quantitatively from \OnlineTSPordering, it is instructive to consider this case where the metric space $M$ is the unit interval $[0,1]$ equipped with the usual Euclidean distance $d$.
Assuming both $0$ and $1$ appear as part of the input, there is a trivial 1-competitive algorithm for \OnlineTSPscheduling$[n,M,d]$ which places $x_i$ at position $p_i=x_i$\footnote{Without assuming that both 0 and 1 are part of the input, a simple guess-and-double strategy would give an $O(1)$-competitive algorithm.}. 
Conversely, when the metric space is the unit interval, \OnlineTSPordering$[n,\eps,M,d]$ is exactly the online sorting problem, \OnlineSorting$[n,\eps]$, for which it is known that any algorithm has competitive ratio $\Omega(\frac{\log n}{\log \log n})$ when $\eps=\Theta(1)$~\cite{aamand2023online}.  

Our techniques also apply to \OnlineTSPordering$[n,\eps,M,d]$ where we obtain the following result, making progress on a question by Bertram~\cite{DBLP:conf/esa/Bertram25}.
This result was obtained concurrently and independently by Azar, Panigrahi, and Vardi~\cite{azar2026beyond}.

\begin{restatable}{theorem}{mainTSP}\label{thm:TSP}
Let $\eps\in (0, 1]$. There exists an algorithm for
\OnlineTSPordering$[n,\eps,M,d]$ with competitive ratio
$O((\log^2 n)/\eps)$ if $\OPT$ is given as part of the input, and competitive ratio $O((\log^3 n)/\eps)$ if $\OPT$ is unknown to the algorithm.
\end{restatable}

It is worth noting that one can obtain a weaker version of~\cref{thm:main}, with competitive ratio $O(\log^3 n)$, by a careful application of the $O((\log^2 n)/\eps)$ competitive algorithm for \OnlineSorting$[n,\eps]$ in~\cite{azar2026nearly}.
The advantage of applying the algorithm of~\cref{thm:TSP:scheduling} for \OnlineTSPscheduling$[n,M,d]$ is not merely the improved competitive ratio, but also that the proof of~\cref{thm:main} becomes very clear. While the online sorting lower bound was useful for the lower bound in~\cite{aamand2023online}, it seems that online TSP scheduling is a more natural problem for designing online geometric packing algorithms.
In the three-dimensional online disk packing problem, discussed below, we use the full power of~\cref{thm:TSP:scheduling}.

\paragraph*{Online translational packing of hyperdisks.}
To illustrate the flexibility of our result on online TSP scheduling, we show how it can be applied to higher dimensional translational online packing problems. When the pieces to be packed can be rotated or are axis-aligned boxes, various results in the literature establish approximation algorithms for packing such convex pieces into minimum volume containers; see, e.g., \cite{diedrich2008approximation, JansenK0ST25, Jansen0LS22, jansen2014new, DBLP:journals/dcg/LassakZ91}. 

In contrast, the translational setting of the problem is much harder.
Even in the offline setting, no known algorithms achieve non-trivial approximation factors for, e.g., \textit{volume packing} of convex polyhedra, where the goal is to pack the pieces in $\mathbb{R}^3$ minimizing the volume of their axis-aligned bounding box.
On the other hand, going beyond axis-aligned boxes, Alt, Cabello, Cheong, Park, and Seiferth~\cite{AltCCPS26} 
presented a constant factor approximation algorithm for translational packing of $d$-dimensional unit hyperdisks into a minimum-volume axis-aligned box in $\mathbb{R}^{d+1}$. When $d=2$, this corresponds to packing 
unit disks (of given normal vectors) in three dimensions. 
To obtain their results, they used techniques very different from methods for packing axis-aligned boxes.
In particular, they showed that the set of disks defines a certain metric TSP problem whose solution yields a constant factor approximation algorithm for packing the disks.
In this paper, we show that our algorithm for online TSP scheduling can be used to produce an online packing of disks by translation with only an $O(\log^2 n)$ blowup in volume used.
Concretely, we consider volume packing, but with minor modifications, our techniques can also be used to obtain analogous results for, e.g., strip packing into $[0,3]^{d}\times [0,\infty)$.

\begin{restatable}{theorem}{mainDisks}\label{thm:packing:3d:disks}    For every $d\in \mathbb{N}$, there is a deterministic online algorithm for packing a sequence of $n$ pairwise non-parallel 
unit hyperdisks in $\mathbb R^{d+1}$ into an axis-aligned box of minimum volume with competitive ratio $O_d(\log^2 n)$.
    The algorithm does not need to know $n$ in advance.
\end{restatable}

The following theorem shows that it is impossible to generalize the result to a $\polylog n$ competitive algorithm for disks of arbitrary radii.

\begin{restatable}{theorem}{lowerDisks}\label{thm:arbitrary:radius:hyperdisk:lower}
For every fixed $d\ge 2$, every deterministic online algorithm for translational volume packing of $d$-dimensional hyperdisks of arbitrary
radii in $\mathbb R^{d+1}$ has competitive ratio
\[
    \Omega_d\left(n^{\frac{d-1}{d(d+1)}}\right).
\]
\end{restatable}

\subsection{Technical Overview}
Here we outline our techniques used to obtain our main results.

\paragraph*{New Bounds for Online Translational Packing Problems.}
We start with sketching how we use~\cref{thm:TSP:scheduling} to obtain~\cref{thm:main}. Full proofs appear in~\cref{sec:box-packing} and~\cref{sec:other-packings}. Let $\mathcal{I}=(P_1,\dots,P_n)$ be an input sequence of convex polygons, and let $A$, $W$, $H$ denote, respectively, the total area of the pieces in $\mathcal{I}$, and the maximum width and height of any piece in $\mathcal{I}$. Here the \textit{width} and \textit{height} of a piece $P$ are,  respectively, the maximum distance between $x$-coordinates and $y$-coordinates of any two points of $P$. First, one can show through standard arguments that one can replace each arriving piece $P$ by a slightly larger piece $P'\supset P$ such that $P'$ is a parallelogram with a pair of horizontal edges and a height that is a power of two. 
Moreover, the area, width, and height of~$P'$ are within constant factors of the area, width, and height of $P$. It therefore suffices to design an online algorithm for packing such parallelograms. We partition these parallelograms according to height, letting~$\mathcal{I}_h$ denote the arriving pieces that have height $h$ (which is a power of two). As a subroutine, we use an online algorithm to pack the arriving elements of $\mathcal{I}_h$ into a strip $S_h:=[0,\infty)\times [0,h]$. Note that there is a natural metric~$d_h$ on $\mathcal{I}_h$ such that for $P_1,P_2\in \mathcal{I}_h$, $d_h(P_1,P_2)$ is the minimum possible distance between the centers of $P_1,P_2$ when packed in $S_h$. We can thus use the algorithm of~\cref{thm:TSP:scheduling} to pack the arriving pieces of $\mathcal{I}_h$ into $S_h$. If this algorithm schedules $P$ at $x\in [0,\infty)$, we pack $P$ with its center at $(x,\frac{h}{2})$\footnote{A small caveat is that the pieces may extend beyond the leftmost boundary of the strip, but this issue turns out to be easy to address.}.  The requirement that $|p_i-p_j|\geq d_h(P_i,P_j)$ in \OnlineTSPscheduling$[n,M,d]$ exactly guarantees that the pieces are non-overlapping, and we obtain a packing that within polylogarithmic factors is that of the minimum TSP cost of $\mathcal{I}_h$. Let $A_h$ denote the total area of the pieces of $\mathcal{I}_h$ and $W_h$ denote the width of the largest piece in $\mathcal{I}_h$. By sorting the pieces in $\mathcal{I}_h$ by slope, it is easy to see that a shortest TSP path costs $O(A_h/h+W_h)$, so we obtain a strip packing into $S_h$ of cost $O( (A_h/h+W_h)\log^2 |\mathcal{I}_h|)$, or equivalently, such that the area of the occupied strip is $O( (A_h+hW_h)\log^2 |\mathcal{I}_h|)$

If all pieces had height $1$, this would give~\cref{thm:main}.
However, all pieces could be of a tiny height $h$, and then it is necessary to split $S_h$ into boxes that can, e.g., be stacked on top of each other. To capture what we want of this splitting, we introduce the notion of \emph{box packing}, a problem that might be of independent interest. In box packing, as pieces arrive, an algorithm can either pack them into existing boxes or open new axis-aligned bounding boxes of arbitrary width and height for the packing. Our goal is to minimize the total area of opened boxes, while keeping the maximum width and height of any opened box within a constant factor of the maximum width and height of the  input pieces. 
By using a guess-and-double strategy for $W_h$, we show that each strip $S_h$ can be split into boxes online, without increasing the area used by more than a constant factor and such that the width of any box is $O(W_h)$, and the height of any box is $h$. Doing so for each height class $h$, we obtain a box packing algorithm such that any opened box has width $O(W)$ and height $O(H)$ and such that the total area of the boxes is 
\[
O\left(\sum_h(A_h+hW_h)\log^2 |\mathcal{I}_h|\right)=O\left((A+WH)\log ^2n\right) .
\]

The reason the box packing result is interesting in its own right is that it can be combined in a black-box fashion with known algorithms for strip packing, perimeter packing, and area packing of axis-aligned rectangles. Focusing here on strip packing, there are standard constant factor competitive algorithms for online strip packing of \textit{axis-aligned rectangles}~\cite{baker1983shelf,coffman1980performance}. If the total area of the rectangles is $A_R$ and the maximum width of a rectangle is $W_R$, these algorithms guarantee a packing of cost $O(W_R+A_R)$. We feed the boxes initialized to this algorithm, and thus incur cost
\[
O(W+(A+WH)\log^2 n)=O((A+W)\log ^2n),
\]
where we used that $H\leq 1$ in a strip packing instance.
As the cost of any strip packing is $\Omega(A+W)$,~\cref{thm:main}~(\ref{item:0}) follows.
Similarly, applying known algorithms for perimeter packing and square packing of axis-aligned rectangles~\cite{AbrahamsenBeretta20}, we obtain ~\cref{thm:main}~(\ref{item:ii}) and~(\ref{item:iii}).
For the bin packing result, ~\cref{thm:main}~(\ref{item:i}), we instead use a black-box reduction from online strip packing.
We view the strip as covered with overlapping unit squares of the form $[ k\delta, k\delta +1]\times [0,1]$ for $k\in\mathbb{N}$. With the assumption that the arriving pieces have width at most $1-\delta$, there will always be some such square that fully contains an arriving piece in which the piece is packed. 

We present our result on box packing in~\cref{sec:box-packing}, and apply it to obtain~\cref{thm:main} in~\cref{sec:other-packings}.

\paragraph*{Online Metric TSP Scheduling.}
We next discuss the proof of~\cref{thm:TSP:scheduling}, which appears in~\cref{sec:TSP-scheduling}.
Our algorithm uses the idea of~\emph{elementary trees} introduced by Azar, Panigrahi, and Vardi~\cite{azar2026nearly} for their analysis of the online sorting problem, but with several modifications. We first describe their elementary trees.

Let $m=2^H$ for some positive integer $H$. An elementary tree $T$ of size $m$ is a rooted perfect binary tree of height $H$ with $m=2^H$ leaves.
We say that the root node $r$ of $T$ has height $H$ and if $v$ is a non-leaf node of height $h$, we say that its two children have height $h-1$. In this way, leaves are of height~$0$. Elements from $[0,1]$ can be inserted into leaves of an elementary tree. During an insertion, some of the internal nodes of $T$ become \emph{labeled} with dyadic intervals $I\subset [0,1]$. A node of height $h$ can be labeled with a dyadic interval of length $2^{h-H}$. To insert an element $x\in [0,1]$, we look for the leftmost leaf $v$ such that for any node with any label $I$ on the unique path from $r$ to $v$, it holds that $x\in I$. Then $x$ is inserted at $v$ and all unlabeled nodes on the $r$-to-$v$ path are labeled with the dyadic interval of the appropriate length containing $x$. It is possible that one cannot insert an element $x$ at any node, and in this case we open a new elementary tree, always attempting to insert a new element in a leftmost possible such tree. Azar, Panigrahi, and Vardi~\cite{azar2026nearly} prove the following result.

\begin{lemma}
\cite[Lemma~3.5]{azar2026nearly}
\label{lemma:elementary:trees} 
At any time during the insertion process, the number of unused leaves (over all opened elementary trees) is at most $H2^H$.
In particular, after initializing $2H$ elementary trees, one can always insert at least $n=H2^H$ elements.
\end{lemma}

They then showed that letting the leaves of each tree in DFS order correspond to consecutive entries of the array in the online sorting problem gives a low cost solution for the online sorting problem. The lemma above guarantees that picking $H$ such that $H2^H=\eps n$, all $n$ elements can be inserted. In particular, they only need to open $n(1+\eps)/2^H=O(H/\eps)=O((\log n)/\eps)$ such trees.

We now discuss the core challenges in generalizing this idea to arbitrary metric space for the online TSP scheduling problem. In this discussion, for simplicity, we assume that $n$ and  $\OPT$ are known. The case where they are not can be handled relatively easily by guess-and-doubling tricks. By rescaling, we may assume that $\OPT=1$.

\underline{\textit{Node labels}:} A natural attempt of a labeling scheme for general metric spaces, is to label internal nodes with elements of the metric. When an element $x$ is inserted at some leaf $v$, all unlabeled nodes on the $r$ to $v$ path get labeled by $x$. Supposing that $u$ is a node of height $h$ with a label $x$, we think of the label as the radius $2^{h-H}$ ball $B(x,2^{h-H})$. For an unlabeled leaf $v_0$, let $v_0,v_1,\ldots,v_H$ be the path from this leaf to the root, where $v_h$ has height $h$. The \textit{labeled depth} of the path
is the number of labeled nodes on it. The path is feasible for $x$ if
either no node on the path is labeled, or, letting $v_h$ be the deepest
labeled node on the path and $y$ its label, we have
\[
    d(x,y)\le 2^{h-H-1}.
\]
If this is the case, then $x$ could be inserted at $v_0$.

\underline{\textit{How to prioritize feasible paths}:} The proof of~\cref{lemma:elementary:trees} proceeds by bounding the number of \textit{partial nodes} at a given height $h$, where a partial node is a node where exactly one of its children is labeled. Essentially, no two such height $h$ partial nodes $v_1,v_2$, with say $v_1$ to the left of $v_2$, can have the same label. Otherwise, the element labeling $v_2$ could have been inserted under the rightmost child of $v_1$, contradicting that we always pick the leftmost feasible path. Thus the number of partial nodes can be at most $2^{H-h}$. In our setting, we would like to argue that for any two partial nodes at height $h$ with labels $x,y$, it holds that $d(x,y)\geq 2^{h-H}$. If there were more than $2^{H-h}+1$ such partial nodes, then this is a witness that the shortest TSP-path has length $>1$, a contradiction. Unfortunately, the bound $d(x,y)\geq 2^{h-H}$ does not hold if we use the leftmost feasible path when inserting an element. For the elementary trees in~\cite{azar2026nearly}, the argument crucially relies on the dyadic labels along a path being \textit{nested}. To see the problem, in the metric setting, let $v_1$ and $v_2$ be labeled by $x_1$ and $x_2$. Even if $d(x_1,x_2)<2^{h-H}$, if $x_2$ is inserted before $x_1$, we cannot conclude that at the time $x_2$ was inserted, there was a feasible path through $x_1$ --- without the nesting, $x_2$ lying in the label of $x_1$, does not imply that $x_2$ lies in labels of the ancestors of $x_1$. All attempts to resolve this issue solely by changing the labeling or the feasibility definition, run into the same issue. 

To overcome this issue, we prioritize inserting elements at large \textit{labeled depth}, dispensing entirely with the idea of going as far left as possible. With the above notation, it is easy to see that $d(x_1,x_2)>2^{h-H}$. If not, the insertion of the latest inserted element $x_i$ caused the labeling of $v_i$, which means that it could have been inserted under $v_{1-i}$ along a feasible path and at greater labeled depth, a contradiction. 

\underline{\textit{Scheduling points in $[0,\infty)$}:} The above argument gives us an analogue of~\cref{lemma:elementary:trees}. In particular, we can insert all $n$ arriving points in $O(\log n)$ {\em metric elementary trees}. It remains to describe how points are scheduled in $[0,\infty)$. For this, we let the leaves $v$ of the trees corresponding to points $x_v\in [0,\infty)$. Label the leaves of a tree $T$ in DFS order $v_1,\dots, v_{2^H}$. For two nodes $v_i,v_{i+1}$ with a least common ancestor of depth $h$, we ensure that  $x_{v_{i+1}}-x_{v_i}=2^{h-H+2}$. When trees are inserted, they are pushed as far left as possible on $[0,\infty)$ with the requirement that we leave space $1$ on $[0,\infty)$ between two trees. Summing over all depths and all trees, the total space used is $O(\log^2 n)$. Finally, by a geometric sum argument, if $x_1,x_2$ are inserted with least common ancestor $v$ of height $h$ labeled with some $y$, then by the feasibility of paths, $d(x_i,y)\leq 2^{h-H+1}$, for $i=1,2$, so $d(x_1,x_2)\leq 2^{h-H+2}$ implying that we are not violating the online TSP scheduling constraints. 

The elementary-tree framework also yields our result for
\OnlineTSPordering\ (\cref{thm:TSP}); since this is not needed for the packing applications, we
defer the details to \cref{sec:TSP}.

\paragraph{Packing disks.}
Finally, in \cref{ssec:disks}, we apply \cref{thm:TSP:scheduling} to obtain \cref{thm:packing:3d:disks} following the approach of
Alt, Cabello, Cheong, Park, and Seiferth~\cite{AltCCPS26}. 
For a fixed direction $s$, they define the \textsc{DiskStabbing} problem: the centers of the hyperdisks must lie on a line of direction $s$ and the goal is to minimize the length of the occupied segment. For two hyperdisks $D, D'$, let $d_s(D, D')$ be the separation needed along this line so that the two translated hyperdisks are non-overlapping. The key idea is that $d_s$ is a metric, hence \textsc{DiskStabbing} is a metric TSP path instance.

They then reduce volume packing to $O(d)$ of these stabbing problems. A hyperdisk with unit normal vector $u=(u_1,\ldots,u_{d+1})$ is assigned to a coordinate direction $e_i$ with maximum $|u_i|$. For each resulting
class, the optimum stabbing length in direction $e_i$ is within a (dimension-dependent) constant factor of the optimum axis-aligned bounding-box volume. After computing the stabbing order, the stabbing line is cut into unit-length portions, and these are rearranged into an axis-aligned container.

Our online algorithm uses the same idea, replacing only the offline TSP-path step. For each coordinate direction $e_i$, we can run the online TSP scheduling algorithm on the metric $d_{e_i}$. When a hyperdisk $D$ arrives, the algorithm assigns it a coordinate on the stabbing line. The scheduling constraints $|p(D) - p(D')| \geq d_s(D, D')$ then guarantee that the corresponding hyperdisks are disjoint and the total span of the stabbing line is within $O_d(\log^2 n)$ factor of the optimum stabbing length in that direction.

We partition each stabbing line into unit intervals and open one constant-size hypercube for every such interval, placing the hyperdisks into these hypercubes and stacking the hypercubes online. Now we can compare this solution to the minimum value of an axis-aligned bounding box $\OPT$:
Since there are $d+1$ classes and each stabbing optimum is bounded by $O_d(\OPT)$, the number of hypercubes and hence used volume is $O_d(\log^2 n)$, giving the $O_d(\log^2 n)$-competitive algorithm.

\subsection{Preliminaries}
\paragraph{Competitive ratios.}
We use the standard notion of competitive ratio for online minimization problems. For such a problem, let $\mathcal{I}$ be an input instance with $|\mathcal{I}|=n$ online updates. For a deterministic online algorithm $\mathcal{A}$, we write $\mathcal{A}(\mathcal{I})$ for the cost of the solution
produced by $\mathcal{A}$ on $\mathcal{I}$, and we write $\OPT(\mathcal{I})$ for the minimum cost of an offline solution. 
The competitive ratio of $\mathcal{A}$ on inputs of size $n$ is
\[
\rho_{\mathcal{A}}(n)
  :=
  \sup_{|\mathcal{I}|=n}
      \frac{\mathcal{A}(\mathcal{I})}{\OPT(\mathcal{I})}.
\]
We say that $\mathcal{A}$ is $O(f(n))$-competitive if $\rho_{\mathcal{A}}(n)=O(f(n))$. 

\paragraph*{Two-dimensional packing problems.}
The input to the two-dimensional packing problems considered in this paper is a
finite sequence
\[
    \mathcal{I}=(P_1,\ldots,P_n)
\]
of convex polygons in the plane. In the online setting, the polygons are revealed one by one in this order. When $P_i$ arrives, the algorithm must pack $P_i$ by translation (formally, pick a vector $t_i$ such that the copy of $P_i$ is placed at $P_i+t_i=\{p+t_i\mid p\in P_i\}$). Rotations and reflections are not allowed. A packing is feasible if the interiors of the translated polygons are pairwise disjoint.

We consider different packing problems with different objectives to minimize.
\begin{itemize}
\item \emph{Strip packing}: The pieces of $\mathcal{I}$ are packed in the strip $S=[0,\infty)\times[0,1]$. The cost of a packing is the largest $x$ coordinate of any occupied point of $S$. 
\item \emph{Bin packing}: The pieces are packed into an infinite sequence of axis-aligned unit squares. The cost is the number of squares containing a piece.
\item \emph{Perimeter packing}: The pieces are packed in the plane. The cost is the perimeter of their minimal axis-aligned bounding box. 
\item \emph{Area Packing}: The pieces are packed in the plane. The cost is the area of their minimal axis-aligned bounding box.
\end{itemize}
\paragraph*{Unit hyperdisk packing.}
For a fixed integer $d \geq 2$, the input is a finite sequence $\mathcal D = (D_1, \dots, D_n)$ of unit hyperdisks in $\mathbb R^{d+1}$.
A unit hyperdisk is a $d$-dimensional unit ball contained in a $d$-dimensional affine subspace. The hyperdisks are packed by translation only, and we say a packing is feasible if the (relative) interiors of the translated disks are pairwise disjoint.

The cost of a packing of the hyperdisks is the volume of a minimal axis-aligned bounding box, and for an instance $\mathcal D$, we write $\OPT_{\mathrm{AABB}}(\mathcal D)$ for the minimum possible cost.

\section{Online TSP Scheduling}\label{sec:TSP-scheduling}
In this section, we  prove~\cref{thm:TSP:scheduling}. We start with the following theorem, where we assume that upper bounds for $n$ and $\OPT$ are known in advance.

\begin{lemma}\label{thm:online:tsp:scheduling:known:opt}
Assume that an upper bound $N\ge n$ on the number of input points and a value $\Delta\ge \OPT$ are given. There is a polynomial-time algorithm for \OnlineTSPscheduling$[n,M,d]$ that produces a schedule of cost $O(\Delta\log^2 N)$. In particular, if $N=n^{O(1)}$ and $\Delta=O(\OPT)$, the competitive ratio is $O(\log^2 n)$. 
\end{lemma}

\begin{proof}
If $\Delta=0$, then all input points have pairwise distance zero, and we assign all of them time $0$. Thus assume $\Delta>0$. By scaling distances and times by $\Delta$, we may assume $\Delta=1$. Then the optimum of the whole input is at most $1$, and in particular any two points of the input have distance at most $1$.

We may additionally assume that $N$ is  larger than a sufficiently large constant $C$. Indeed, if $N\leq C$, we may schedule the $i$'th arriving point at $i-1$. The distance between any pair of points is at most the minimal TSP cost which is at most $1$ by assumption, so this is a feasible schedule incurring cost $n\leq N\leq C=O(1)$.
Let $H\ge 1$, be maximal such that $(H+1)2^H\le N$, and let  $\ell:=\left\lceil \frac{2N}{2^H}\right\rceil$.
Our algorithm uses  $\ell$ elementary trees as in~\cite{azar2026nearly}, each of height $H$, but with a different labeling and insertion procedure. To recap, an elementary tree is a rooted perfect binary tree with $2^H$ leaves, ordered from left to right in DFS order. We say that the root has height $H$ and that the children of a node of height $h$ have height $h-1$ so that leaves have height $0$. During the insertion process, the algorithm labels certain internal nodes of the tree with input points, as described below and depicted in~\cref{fig:elementarytree}.

\begin{figure}[htb]
    \centering
    \includegraphics[page=1]{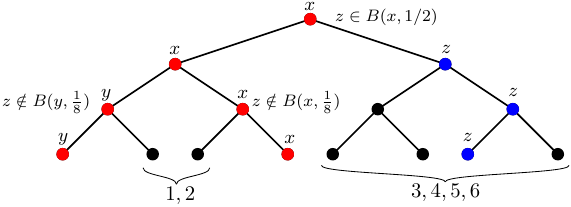}
    \caption{An insertion of the element $z$ into an elementary tree of height 3 which already contains elements $x,y$. Labeled nodes before the insertion of $z$ are red. The algorithm attempts to insert $z$ at a leaf in the order indicated below the tree according to the maximum labeled depth requirement. Since $z\notin B(x,1/8)$ and $z\notin B(y,1/8)$, the paths to the first two leaves in this order are not feasible. Since $z\in B(x,1/2)$, and the rightmost child of the root is unlabeled, $z$ can be inserted under this child along any root to leaf path. Finally, the unlabeled nodes of this path (in blue) get label $z$.}
    \label{fig:elementarytree}
\end{figure}

For an unlabeled leaf $v_0$, let $v_0,v_1,\ldots,v_H$ be the path from this leaf to the root, where $v_h$ has height~$h$. We define the \textit{labeled depth} of the path as the number of labeled nodes on it. We say that the path is \textit{feasible} for a point $x$ if either no node on the path is labeled, or, letting $v_h$ be the deepest labeled node on the path and $y$ its label, we have \[ d(x,y)\le 2^{h-H-1} . \]

Now we describe insertions. When a point $x$ arrives, our algorithm chooses a feasible path of \textit{maximum} labeled depth among all $\ell$ elementary trees, breaking ties arbitrarily. The point $x$ is inserted at the leaf of the chosen path, and all previously unlabeled nodes on the path are labeled by $x$. If no feasible path exists, the insertion fails. We note the easy invariant that if a node is labeled, then either it is a root or its parent is also labeled.

If a node $v$ at height $h$ is labeled by $y$, then every point inserted in the subtree rooted at $v$ is at distance at most $2^{h-H}$ from $y$. Indeed, let $v=w_h,w_{h-1},\ldots,w_0$ be the path from $v$ to the leaf containing such a point $x$, where $w_r$ has height $r$, and let $y_r$ be the label of $w_r$. Then $y_h=y$ and $y_0=x$, so by the feasibility constraint and the triangle inequality,
\begin{align*}\label{eq:ball-containment}
    d(x,y)\le \sum_{r=0}^{h-1} d(y_r,y_{r+1}) \le \sum_{r=0}^{h-1} 2^{r-H} < 2^{h-H}.
\end{align*}
The following claim follows immediately from this observation and the triangle inequality.
\begin{claim}\label{claim:sched:tree:distance}
If two points $x$ and $x'$ are inserted at leaves of the same elementary tree whose least common ancestor has height $h$, then \[ d(x,x')\le 2^{h-H+1} . \]
\end{claim}

Call a node \emph{partial} if it has exactly one labeled child.

\begin{claim}\label{claim:sched:partial:separation}
Let $v$ and $v'$ be distinct partial nodes of height $h$, possibly in different elementary trees. If $y$ and $y'$ are their labels, then \[ d(y,y')> 2^{h-H-1} . \]
\end{claim}

\begin{proof}[Proof of Claim.]
Assume for contradiction that $d(y,y')\leq 2^{h-H-1}$. Let $t$ and $t'$ be the times at which $v$ and $v'$ were labeled, respectively, and assume without loss of generality that $t<t'$. Since $v$ is partial at the end of the process, its unlabeled child was still unlabeled just before time $t'$. Hence, just before inserting $x_{t'}=y'$, there was a root-to-leaf path through this unlabeled child whose deepest labeled node was $v$, thus having labeled depth $H-h+1$. This path was feasible for $y'$ by the assumed inequality.

On the other hand, the insertion at time $t'$ caused the labeling of the node $v'$, which was unlabeled just before time $t'$. The labeled depth of the feasible path used for the insertion at time $t'$ was therefore at most $H-h$.
This contradicts the rule that the algorithm chooses a feasible path of maximum labeled depth.
\end{proof}

Call a subset $S$ of the input $\{x_1,\dots,x_n\}$ $\delta$\textit{-separated} if any distinct $x,x'\in S$ have $d(x,x')>\delta$.
For every $\delta>0$, any $\delta$-separated subset of the input has size at most $1/\delta+1$. Indeed, order the elements of $S$ as they appear along an optimal traveling-salesperson path of length at most $1$, and sum the distances between consecutive points in $S$. 

By \cref{claim:sched:partial:separation}, the labels of partial nodes at any fixed height $h$ are $2^{h-H-1}$-separated, so the number of partial nodes at height $h$ is at most \[ 2^{H-h+1}+1 . \]
Every empty leaf in a non-empty elementary tree lies in the subtree rooted at the unlabeled child of a unique partial node. A partial node at height $h$ accounts for at most $2^{h-1}$ such empty leaves. Therefore the number of empty leaves in non-empty elementary trees is at most \[ \sum_{h=1}^H (2^{H-h+1}+1)2^{h-1} \le (H+1)2^H \le N . \]
Since there are at least $2N$ leaves in total, the algorithm is able to insert all $n$ arriving points. Indeed, if it failed before inserting all $n$ points, then no empty tree would remain, so all empty leaves would lie in non-empty trees. But then the number of empty leaves in non-empty trees would be $\geq 2N-(n-1)>N$, a contradiction.

It remains to assign visit times. In each elementary tree, let $v_1,\ldots,v_{2^H}$ be its leaves in DFS order. For $j<2^H$, let $h_j$ be the height of the least common ancestor of $v_j$ and $v_{j+1}$, and set \[ q_1=0, \qquad q_{j+1}-q_j=2^{h_j-H+1} . \]
There are exactly $2^{H-h}$ adjacent-leaf boundaries whose least common ancestor has height $h$. Hence \[ q_{2^H} =\sum_{h=1}^H 2^{H-h}2^{h-H+1} =2H . \]
In particular, $0\le q_j\le 2H$ for every $j$. 
Intuitively, $q_j$'s denote the local position of a point within its elementary tree. 
If a point is inserted at leaf $v_j$ of the $k$th elementary tree, we assign it the irrevocable visit time \[ p=(k-1)(2H+1)+q_j . \]

We now verify feasibility. Consider two points inserted into the same elementary tree, and suppose their leaves have least common ancestor of height $h$. Then the two points are scheduled with distance at least $2^{h-H+1}$ in $[0,\infty)$. Moreover, by \cref{claim:sched:tree:distance}, their metric distance is at most $2^{h-H+1}$, so the TSP scheduling requirement is satisfied. Now consider two points in distinct elementary trees, say the $k$th and $k'$th trees with $k<k'$. Their assigned times differ by at least \[ (k'-k)(2H+1)-2H \ge 1 . \]
This is at least their metric distance, since after scaling the input has diameter at most $1$.

By maximality of $H$, under the assumption that $(H+1)2^H\leq N$, we have $H=O(\log N)$ and $2^H=\Omega(N/\log N)$, and therefore $\ell=O(\log N)$. The cost incurred by the algorithm is therefore at most 
\[
\ell(2H+1)=O(\log^2 N) . 
\]
Rescaling by $\Delta$ proves the theorem.
\end{proof}

Next, we use the previous theorem to prove a similar bound in the case where $\OPT$ is unknown.

\begin{lemma}\label{thm:online:tsp:scheduling:unknown:opt}
When $n$ is known in advance, there exists an algorithm for
\OnlineTSPscheduling$[n,M,d]$ with competitive ratio
$O(\log^2 n)$. 
\end{lemma}
\begin{proof}
For a prefix $X_t=\{x_1,\ldots,x_t\}$, let $\OPT_t$ denote the optimum TSP-path length for $X_t$. Upon arrival of $x_t$, we compute in polynomial time a value $A_t$ such that
\[
    \OPT_t \le A_t \le 2\OPT_t,
\]
for instance by the standard 2-approximation algorithm for TSP based on doubling the minimum spanning tree (MST) and then shortcutting.  

The algorithm proceeds in phases, where in each phase, we keep a current estimate of $\OPT_t$. As long as $A_t=0$, assign $x_t$ visit time $0$. When the first point with $A_t>0$ arrives, we start a phase with estimate $\Delta_1=A_t$. In general, if the current phase has estimate $\Delta_i$ and $A_t>\Delta_i$, we close the current phase and start a new phase with
\[
    \Delta_{i+1}=2^k\Delta_i,
\]
where $k\ge 1$ is the minimum integer such that $2^k\Delta_i\ge A_t$. If $x_t$ caused a new phase $i$, $x_t$ is inserted during phase $i$.

For each phase $i$, we run the algorithm from \cref{thm:online:tsp:scheduling:known:opt} with $N=n$ and $\Delta=\Delta_i$. When a new phase $i$ starts, we translate the resulting visit times and schedule every arrival in phase $i$ to the right of every arrival in phase $i-1$, leaving a gap of exactly $\Delta_i$. Note that for an arrival time $t$ in phase $i$, we always have that $\OPT_t\leq A_t\leq \Delta_i$, and in particular, the maximum pairwise distance between a pair of points in the instance is at most $\Delta_i$, so a gap of size $\Delta_i$ suffices for feasibility between the phases. Feasibility within a phase follows from \cref{thm:online:tsp:scheduling:known:opt}.

We also need to check that the algorithm of \cref{thm:online:tsp:scheduling:known:opt} is always called with a valid upper bound. Clearly, $N=n$ is a valid upper bound on the number of arriving elements. Moreover, at any time $t$ in phase $i$, we have that 
$\OPT_t\leq A_t\leq \Delta_i$, since this holds by the start of the phase by the choice of $k$, and since we proceed to the next phase when it is no longer true.

It remains to bound the cost. By \cref{thm:online:tsp:scheduling:known:opt}, phase $i$ has length $O(\Delta_i\log^2 n)$, and the additional gap contributes only $\Delta_i$. The $\Delta_i$ terms are (at least) geometrically increasing, and for all $i$, we have $\Delta_i<2A_t\le 4\OPT$, by the minimality of $k$, so $\sum_i \Delta_i=O(\OPT)$. Therefore the total cost is
\[
    \sum_i O(\Delta_i\log^2 n)
    = O(\OPT\log^2 n).
\]
The algorithm runs in polynomial time since each $A_t$ is computed by a polynomial-time $2$-approximation and the subroutine from \cref{thm:online:tsp:scheduling:known:opt} runs in polynomial time.
\end{proof}

We are now ready to prove~\cref{thm:TSP:scheduling}.

\mainTSPscheduling*

\begin{proof}
Let $N_0=2$ and $N_{i+1}=N_i^2$. We split the input into phases, where phase $i$ contains the next at most $N_i$ points. Within phase $i$, we run the algorithm from \cref{thm:online:tsp:scheduling:unknown:opt} with input size parameter $N_i$, stopping it if the phase ends earlier, in which case, we have seen all $n$ input points.

We schedule the phases consecutively, scheduling the first arrival from the first phase at $0$. Let $z_i$ be the first point of phase $i$. Before starting phase $i>0$, we leave an empty gap of length
\[
    D_i=\max\{d(z_i,x)\mid x \text{ arrived before phase } i\}.
\]
The first point in each phase is assigned {\em local} time $0$, as in the proof of \cref{thm:online:tsp:scheduling:unknown:opt}, and then the whole local schedule of the phase is translated to start after the gap.

Feasibility inside each phase follows from \cref{thm:online:tsp:scheduling:unknown:opt}. For feasibility between phases, let $y$ be a point in phase $i$ and let $x$ be any earlier point. If $q_y$ is the local time of $y$, then $q_y\ge d(y,z_i)$, since $z_i$ has local time $0$. Hence the difference between the scheduled time of $y$ and $x$ is 
\[
    p_y-p_x\ge D_i+q_y\ge d(x,z_i)+d(y,z_i)\ge d(x,y),
\]
so the scheduling constraint also holds between different phases.

It remains to bound the cost. Let $\OPT_i$ be the optimum TSP-path length of the points in phase $i$. Then $\OPT_i\le \OPT$, and also $D_i\le \OPT$. 
Let $R$ be the last nonempty phase. Then
 $\log N_R = O(\log n)$.
Since $\log N_i=O(2^i)$, we have by \cref{thm:online:tsp:scheduling:unknown:opt}, the cost of the algorithm is
\[
    \sum_{i=0}^R O(\OPT_i\log^2 N_i)+O(R\OPT)
    =O(\OPT\log^2 n).
    \qedhere
\]
\end{proof}
We will need the following technical observation about the algorithm of~\cref{thm:TSP:scheduling}, which follows immediately from its proof.
\begin{observation}\label{observation:first:at:zero}
The algorithm in~\cref{thm:TSP:scheduling} schedules the first arriving point at $0$.
\end{observation}

\section{Box Packing}\label{sec:box-packing}
In this section, we present an algorithm for an auxiliary problem, \emph{box packing}. 
\begin{definition}[Online Box Packing]
\label{def:boxpacking}
A sequence of $n$ convex polygons $\mathcal{I}=(P_1,\dots,P_n)$ arrive online. Upon arrival of a new piece, the algorithm can open an axis-aligned rectangular box of arbitrary dimensions, or pack the arriving piece into any previously opened box. 
\end{definition} 
To state our result on box packing, we first introduce a bit of notation.
For a polygon $P \subseteq \mathbb{R}^2$, let $\area(P)$ denote its area. For a sequence of polygons $\mathcal{I} = (P_1,P_2,\dots,P_n)$, write $\area(\mathcal{I})=\sum_{i=1}^n \area(P_i)$. For a polygon $P \subseteq \mathbb{R}^2$, define its \emph{width} by 
\[
\width(P)
   := \max\{x : (x,y)\in P\}-\min\{x : (x,y)\in P\},
\]
and its \emph{height} by
\[
\height(P)
   := \max\{y : (x,y)\in P\}-\min\{y : (x,y)\in P\}.
\]

Our box packing result is as follows.
\begin{lemma}\label{lemma:box:packing}
There exists an online algorithm for box packing $n$ convex polygons, $\mathcal{I}=(P_1,\dots,P_n)$. Letting $W=\max_i \width (P_i)$ and $H=\max_i \height (P_i)$, each box opened has width $O(W)$, height $O(H)$ and the total area of the opened boxes is $O(\log^2n (\area(\mathcal{I})+WH))$.
\end{lemma}

\subsection{Reduction to Dyadic Parallelograms}\label{sec:modification}
To prove our box packing result, we first argue that we can employ some fairly standard modifications to the arriving pieces. We say that $P$ is a \emph{horizontal} parallelogram if $P$ is a parallelogram having two horizontal sides, see \cref{fig:horizontalPolygon} (a).
For a horizontal parallelogram $P$, we define the \emph{base} of $P$, $\base(P)$, as the length of either of these horizontal sides. 
We define the \emph{shadow} of $P$ as $\shadow(P)=\width(P)-\base(P)$. Note that shadows are always non-negative. 
We call a real number $x>0$ \emph{dyadic} if $x=2^k$ for some $k \in \mathbb Z$. We call a horizontal parallelogram $P$ \emph{dyadic} if its height is dyadic.

\begin{figure}[htb]
    \centering
    \includegraphics[page=1]{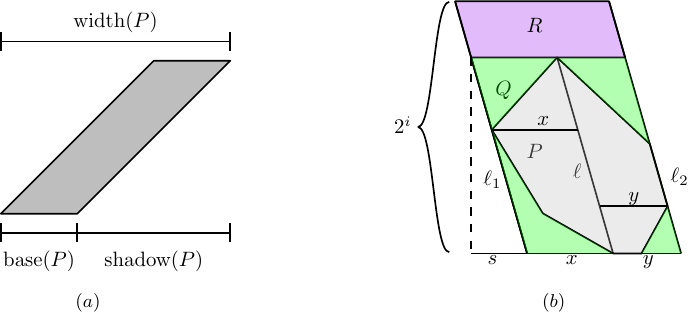}
    \caption{(a) A horizontal parallelogram $P$ with $\base(P)$, $\shadow(P)$, and $\width(P)$. (b) Illustration for \cref{lemma:modification}.}
    \label{fig:horizontalPolygon}
\end{figure}

\begin{lemma}\label{lemma:modification}
For every convex polygon $P$, there exists a dyadic parallelogram $R$ such that 
\begin{enumerate}
    \item $P \subseteq R$,
    \item $\area(R) \leq  4\area(P)$, 
    \item $\width(R) \leq  4\width(P)$, and 
    \item If $k$ is the smallest integer such that $\height(P)\leq 2^k$, then $\height(R)=2^k$.
\end{enumerate}
\end{lemma}
\begin{proof}
Consider \Cref{fig:horizontalPolygon} (b) for an illustration. Let $\ell$ be a line segment between a topmost and bottommost point of $P$. Let $\ell_1$ and $\ell_2$ be obtained by horizontal translates of $\ell$ such that $\ell_1$ is tangent to $P$ on the left and $\ell_2$ is tangent to $P$ on the right. Let $Q$ be the convex hull of $\ell_1\cup\ell_2$, which is a horizontal parallelogram.
It is easy to check that $\area(Q)\leq 2\area(P)$, as $P$ contains the convex hull of $\ell$ and the intersection points of $\ell_1$ and $\ell_2$ with $P$. For the width, consider~\Cref{fig:horizontalPolygon} (b), and note that the width of $Q$ is $s+x+y$, where $s$ is the shadow of the parallelogram $Q$. It is clear that $\max(x,y,s)\leq \width(P)$, so 
\[
\width(Q)=s+x+y\leq 3\width(P).
\]

We can obtain the final dyadic parallelogram $R$ with the desired properties by at most doubling the height of $Q$. It is easy to check that the area increases by at most a factor of two and the width by at most $s\leq \width(P)$ by this operation.
\end{proof}

\subsection{Algorithm for Box Packing}
In this section, we prove~\cref{lemma:box:packing}.
We start with the following lemma about strip packing that assumes that all arriving pieces are dyadic parallelograms of height one.
\begin{lemma}\label{lemma:strip:same:height}
There exists an $O(\log^2 n)$-competitive online algorithm for strip packing $n$ unit-height horizontal parallelograms $\mathcal{I}=\{P_1,\dots,P_n\}$. The algorithm need not know $n$ in advance.
\end{lemma}
\begin{proof}
There is a natural metric on $\{1,\dots, n\}$ defined by these parallelograms. Namely, if $i\neq j$ and $P_i,P_j$ have bases $b_i,b_j$ and shadows $s_i,s_j$, then we can define
\[
d(i,j)=\begin{cases}
0 & \text{if $i=j$},\\
\frac{b_i+b_j+|s_i-s_j|}{2}&  \text{if $i\neq j$ and both $P_i,P_j$ have nonnegative slope or both have negative slope,}\\
\frac{b_i+b_j+s_i+s_j}{2}& \text{otherwise}.
\end{cases}
\]
This metric measures the minimum possible distance between the centers of $P_i$ and $P_j$ in any packing of them into the strip. It is easy to check that this is indeed a metric.
Now run the algorithm of~\cref{thm:TSP:scheduling} on inputs $\{1,\dots,n\}$ considered as points of this metric space. Assume that the algorithm schedules $i$ at $t_i\in[0,\infty)$.  \Cref{observation:first:at:zero} implies $t_1=0$. If $P_1$ has base $b_1$ and shadow $s_1$, we have to make sure that it is packed so that it does not extend beyond the left wall of the strip. This can be done by defining $c=\frac{b_1+s_1}{2}$ and packing $P_1$ at $(c,1/2)$. In general, our algorithm packs $P_i$ with its center at $(c+t_i,1/2)$. The feasibility of this packing follows from the definition of the metric and the constraints of the online TSP scheduling problem. To bound the competitive ratio, note that the cost of a traveling salesperson path of $\{1,\dots, n\}$ in this metric is $O(\area(\mathcal{I})+\max_i \width (P_i))$, which can be seen by sorting the pieces $P_i$ by slope. The cost of the TSP scheduling solution returned by the algorithm of~\cref{thm:TSP:scheduling} is thus $O(\log^2 n(\area(\mathcal{I})+\max_i \width (P_i)))$. The cost of the packing produced is at most $\max_i \width (P_i)$ larger than this, accounting for the fact that the leftmost and rightmost pieces can each extend $\frac{1}{2}\max_i \width (P_i)$ beyond the centers of the pieces. The desired result follows since any strip packing algorithm must incur cost at least $\max(\area(\mathcal{I}),\max_i \width (P_i))$.
\end{proof}

We first use this lemma to design an algorithm for box packing in the special case where all pieces are unit-height dyadic parallelograms. 

\begin{lemma}\label{lemma:box:packing:helper}
There exists an online algorithm for box packing $n$ horizontal parallelograms of height one, $\mathcal{I}=(P_1,\dots,P_n)$ such that, letting $W=\max_i \width (P_i)$, each box opened has width at most $4W$, height $1$, and the total area of the boxes is $O(\log^2n (\area(\mathcal{I})+W))$.
\end{lemma}
\begin{proof}
If we are given a value $W_0$ with $W\leq W_0\leq 2W$ as part of the input, we can directly use the strip packing algorithm of~\cref{lemma:strip:same:height}. Namely, cover the strip with boxes of the form $[kW_0,kW_0+2W_0]\times [0,1]$ where $k$ is a nonnegative integer. Every arriving piece packed by the strip packing algorithm will fit entirely in one of these boxes, and we use this box in the box packing algorithm. The total area of the boxes is clearly within a constant factor of the area used by the strip packing algorithm, and thus within $O(\log^2n)$ factor of $\area(\mathcal{I})+W$.

In general, we use a guess-and-double strategy for $W_0$. Let $W^{t}$ be the maximum width of all the pieces seen up to time $t$. We proceed in phases, such that at times $t$ in phase $i$, we have an estimate $W_i$ such that $W^{t}\leq W_i\leq 2W^{t}$. Whenever this constraint is violated, we start a new phase $i+1$ picking $W_{i+1}=2^kW_i$, where $k$ is the minimal integer such that the constraint holds for phase $i+1$. For phase $i+1$, we use completely new boxes. Let $\mathcal{I}_i$ be the pieces of phase $i$. For phase $i$, we use the algorithm above, which guarantees that the total area used for the boxes is $O(\log^2 n (\area(\mathcal{I}_i)+W_i))$. Summing this over all $i$ and using that the $W_i$ increase geometrically gives the desired result.
\end{proof}

We are now ready to prove our main box packing lemma.

\begin{proof}[Proof of~\cref{lemma:box:packing}]
When a polygon $P_i$ arrives, we first replace it by the dyadic horizontal parallelogram $P_i'$ from \cref{lemma:modification}. We then pack $P_i'$ and place $P_i$ by the same translation. By \cref{lemma:modification}, the resulting instance $\mathcal I'=(P_1',\dots,P_n')$ satisfies
\[
    \area(\mathcal I')=O(\area(\mathcal I)),\qquad
    \max_i\width(P_i')=O(W),\qquad
    \max_i\height(P_i')=O(H).
\]
It remains to pack the dyadic horizontal parallelograms. We treat each height $h$ (which is a power of two) separately. Let $\mathcal I'_h$ be the subsequence of pieces of height $h$, and let $W_h=\max_{P\in\mathcal I'_h}\width(P)$. For each height class, we use the box packing algorithm of~\cref{lemma:box:packing:helper} scaled to height $h$. The total area of the boxes in height class $h$ is $O(\log^2 n(\area(\mathcal I'_h)+hW_h))$ by that lemma, so the total area of all boxes is  
\[
\sum_h O(\log^2 n(\area(\mathcal I'_h)+hW_h)).
\]
Since the sets $\mathcal I'_h$ partition $\mathcal I'$, since $W_h\le \max_i\width(P_i')$ for all $h$, and since the distinct nonempty dyadic heights sum to at most twice $\max_i\height(P_i')$, this is
\[
    O(\log^2 n(\area(\mathcal I')
    +\max_i\height(P_i')\max_i\width(P_i'))).
\]
Using the bounds from \cref{lemma:modification}, this is $O(\log^2 n(\area(\mathcal I)+HW)).$
The width and height bounds on the boxes follow directly from~\cref{lemma:box:packing:helper}.
\end{proof}

\section{Main Online Planar Packing Results}
\label{sec:other-packings}
In this section, we prove~\cref{thm:main}, which we repeat here for convenience.

\main*

\subsection{Strip Packing}
To prove~\cref{thm:main}~(\ref{item:0}), we will use, as a subroutine,  the following result for online strip packing of axis-aligned rectangles:
\begin{lemma}[\cite{baker1983shelf,coffman1980performance}]\label{lemma:strip:rectangles}
There exists an algorithm for translational online strip packing of axis-aligned rectangles $\mathcal{I}=(R_1,\dots,R_n)$ into a strip of height~$1$ and unbounded length that produces a packing of cost
\[
O(\area(\mathcal{I})+ \max_{1\le i\le n} \width(R_i)).
\]
The algorithm need not know $n$ in advance.
\end{lemma}
Indeed, Coffman, Garey, Johnson, and Tarjan~\cite{coffman1980performance} establish the offline bound
\[
\OPT(\mathcal{I})\leq 2\area(\mathcal{I}) +\max_{1\le i\le n} \width(R_i),
\]
and Baker and Schwarz~\cite{baker1983shelf} give a constant-competitive online
algorithm for strip packing.

We are now ready to prove the result from~\cref{thm:main} regarding strip packing.

\begin{proof}[Proof of~\cref{thm:main}~(\ref{item:0})]
Let us denote the strip packing algorithm of~\cref{lemma:strip:rectangles} by $\mathcal{A}$.
We run the box packing algorithm of~\cref{lemma:box:packing} on the arriving pieces $\mathcal{I}=(P_1,\dots, P_n)$, and feed the created boxes to $\mathcal{A}$. Let $B_1,\dots,B_m$ be the created boxes where $m\leq n$. Let $A=\sum_i \area(P_i)$, $W=\max_i \width(P_i)$, and $H=\max_i \height(P_i)$. Similarly, let $A_R=\sum_i \area(B_i)$, $W_R=\max_i \width(B_i)$, and $H_R=\max_i \height(B_i)$.~\cref{lemma:box:packing} yields that $A_R=O((A+WH)\log^2n)$, $W_R=O(W)$ and $H_R=O(H)$. Using~\cref{lemma:strip:rectangles}, it follows that the cost of the strip packing produced by $\mathcal{A}$ is 
\[
O(A_R+W_R)=O((A+WH)\log^2n+W)=O((A+W)\log^2 n),
\]
where the last step used that $H\leq 1$ in strip packing. The desired result follows since trivially any packing must incur cost $\max(A,W)$.
\end{proof}

\subsection{Bin Packing}\label{sec:bin-packing}
In this section, we study the online translational bin packing problem.
We proceed by a reduction to strip packing.

\begin{proof}[Proof of~\cref{thm:main}~(\ref{item:i})]
Assume that every input polygon $P$ has horizontal span at most $\delta<1$.
We prove that one can obtain an online bin-packing algorithm $A_\textrm{bin}$ with competitive ratio $\frac{1}{1-\delta} \cdot \log^2 n$.
The idea is to simulate the $O(\log^2 n)$-competitive strip packing algorithm from~\cref{thm:main}~(\ref{item:0}), denoted by $A_{\mathrm{strip}}$, but instead of placing objects in a strip, we place them into bins. The bin to which an object is assigned depends on where $ A_{\mathrm{strip}}$ places it in the strip.

We define $\Delta = 1-\delta$, and for each integer $k\geq 0$, let
\[
    J_k=[k\Delta, (k+1)\Delta] \qquad \text{and} \qquad W_k=[k\Delta, k\Delta+1] \times [0, 1] .
\]
We assign each polygon $P$ that would be placed on the strip the unique index $k$, such that its leftmost $x$-coordinate lies in $J_k$.
This implies $P \subseteq W_k$, because we know its leftmost $x$-coordinate lies in $[k\Delta, (k+1)\Delta]$ and because $P$ has horizontal span at most $\delta$, its rightmost $x$-coordinate is upper-bounded by
\[(k+1)\Delta + \delta = (k+1)(1-\delta)+\delta = k\Delta+1.\]
Hence all polygons assigned to the same index $k$ can be placed into the same unit-square bin. Therefore the number of bins used is at most the number of ``windows'' $W_k$ intersecting the used strip. Let $L$ be the length of the used strip, i.e., the cost of the strip packing. The number of used bins is therefore
$\big\lceil \frac{L}{\Delta}\big\rceil$.
Thus
\[
    A_\textrm{bin}(I) \leq \bigg\lceil\frac{A_\textrm{strip}(I)}{\Delta}\bigg\rceil \leq \frac{A_\textrm{strip}(I)}{1-\delta} +1.
\]
Since any packing into $k$ unit bins yields a strip packing of length $k$, then
$\textrm{OPT}_\textrm{strip}(\mathcal I) \leq \textrm{OPT}_\textrm{bin}(\mathcal I)$.
Thus, the cost is 
\[
     A_\textrm{bin}(I) \leq \frac{1}{1-\delta} \cdot A_\textrm{strip}(\mathcal I) + 1 \leq \frac{\log^2(n)}{1-\delta}\cdot \textrm{OPT}_\textrm{strip}(\mathcal I)+1 \leq \frac{\log^2(n)}{1-\delta}\cdot \textrm{OPT}_\textrm{bin}(\mathcal I)+1 .
\]
The desired result follows since $\textrm{OPT}_\textrm{bin}\geq 1$.
\end{proof}

\subsection{Perimeter Packing}\label{sec:perimeter-packing}
We now turn our attention to perimeter packing. We use the constant-competitive algorithm of Abrahamsen and Beretta \cite{AbrahamsenBeretta20} for translational perimeter packing of axis-aligned rectangles.

\begin{lemma}[\cite{AbrahamsenBeretta20}, proof of Theorem~1] \label{lemma:perimeter:rectangles} There exists a deterministic online algorithm for translational packing of axis-aligned rectangles $\mathcal R$ such that the perimeter of the bounding box is 
\[O\left(\sqrt{\area(\mathcal R)} + \max \{H_{\mathcal R}, W_{\mathcal R}\}\right),\] 
where  $W_{\mathcal R} = \max_{R \in \mathcal R} \width(R)$ and $H_{\mathcal R} = \max_{R \in \mathcal R} \height(R)$. 
\end{lemma}

We can now prove our main result about perimeter packing.

\begin{proof}[Proof of~\cref{thm:main}~(\ref{item:ii})]
Let $\mathcal I = (P_1,\ldots,P_n)$ be the input, and write \[ W:=\max_i\width(P_i),\qquad H:=\max_i\height(P_i). \]
We run \cref{lemma:box:packing} on $\mathcal I$, and place every opened box using the rectangle packing algorithm from \cref{lemma:perimeter:rectangles}. Let $\mathcal B$ be the set of opened boxes, and define 
\[  W_{\mathcal B}:=\max_{B\in\mathcal B}\width(B),\qquad H_{\mathcal B}:=\max_{B\in\mathcal B}\height(B). 
\]
We get $\area(\mathcal B) = O(\log^2 n (\area(\mathcal I)+HW))$, $H_\mathcal B = \Theta(H)$ and $W_\mathcal B = \Theta(W)$.

Each original polygon is packed inside its assigned box, so this gives a feasible packing of $\mathcal I$. Moreover, by \cref{lemma:perimeter:rectangles}, the perimeter of the produced bounding box is 
\begin{align*}
    O\left(\sqrt{\area(\mathcal B)} + \max\{H_\mathcal B, W_\mathcal B\}\right)
    &= O\left(\log n \sqrt{\area(\mathcal I)+HW} + \max\{H, W\}\right) \\
    &= O\left(\log n (\sqrt{\area(\mathcal I)} + \max\{H, W\})\right)\\
    &= O\left(\log n \cdot \OPT_{\mathrm{per}}(\mathcal I)\right) .
\end{align*}
Here the last step follows from the fact that any solution to perimeter packing of $\mathcal{I}$ must have perimeter at least $\max (4\sqrt{\area(\mathcal{I}}),2H,2W)$

\end{proof}

\subsection{Area Packing}\label{sec:area-packing}

We now consider the online translational packing problem in the plane in
which the cost of a packing is the area of its axis-aligned bounding box.
For an instance $\mathcal I$, let $\OPT_{\mathrm{area}}(\mathcal I)$
denote the minimum possible area of such a bounding box.

We use the following consequence of the analysis of the algorithm by
Abrahamsen and Beretta~\cite{AbrahamsenBeretta20}.

\begin{lemma}[\cite{AbrahamsenBeretta20}, proof of Theorem 9]
\label{lemma:area:rectangles}
There exists an online algorithm for
translational packing of axis-aligned rectangles
$\mathcal R=(R_1,\dots,R_m)$ in the plane such that the area of the
axis-aligned bounding box produced by the algorithm is
\[
\mathcal A_{\mathrm{rect}}(\mathcal R)
=
O\left(
\sqrt m\left(
\area(\mathcal R)
+
\max_{1\le i\le m}\width(R_i)
\cdot
\max_{1\le i\le m}\height(R_i)
\right)
\right).
\]
\end{lemma}

Our main result about area packing is now proved as follows.

\begin{proof}[Proof of~\cref{thm:main}~(\ref{item:iii})] 
Let $\mathcal I=(P_1,\ldots,P_n)$ be the input, and write 
\[ 
    W:=\max_i\width(P_i),\qquad H:=\max_i\height(P_i). 
\] 
Every feasible packing has bounding-box area at least $\area(\mathcal I)$ and at least $HW$. 
Hence $ \area(\mathcal I)+HW=O(\OPT_{\mathrm{area}}(\mathcal I))$.  
We run \cref{lemma:box:packing} on $\mathcal I$, and place every opened box using the rectangle packing algorithm from \cref{lemma:area:rectangles}. Let $\mathcal B$ be the set of opened boxes, and define 
\[ 
W_{\mathcal B}:=\max_{B\in\mathcal B}\width(B),\qquad H_{\mathcal B}:=\max_{B\in\mathcal B}\height(B). 
\] 
We get 
\[ 
\area(\mathcal B)=O(\log^2 n(\area(\mathcal I)+HW)), \qquad H_{\mathcal B}=O(H), \qquad W_{\mathcal B}=O(W). 
\] 
Moreover, $|\mathcal B|\le n$. By \cref{lemma:area:rectangles}, the area of the produced bounding box is \begin{align*} O\left(\sqrt{|\mathcal B|} \left(\area(\mathcal B)+H_{\mathcal B}W_{\mathcal B}\right)\right) &= O\left(\sqrt n \left(\log^2 n(\area(\mathcal I)+HW)+HW\right)\right)\\ &= O\left(\sqrt n\log^2 n(\area(\mathcal I)+HW)\right)\\ &= O\left(\sqrt n\log^2 n\cdot \OPT_{\mathrm{area}}(\mathcal I)\right), \end{align*} 
where the last step follows from the fact that any solution to area packing of $\mathcal{I}$ must have area  at least $\max (\area(\mathcal{I}),HW)$.
Each original polygon is packed inside its assigned box, so this gives a feasible packing of $\mathcal I$. 
\end{proof}

\section{Packing Hyperdisks}
\label{ssec:disks}
We finish with an application of online TSP scheduling to the higher-dimensional packing problem introduced by Alt, Cabello, Cheong, Park, and Seiferth~\cite{AltCCPS26}. 
The offline algorithm from~\cite{AltCCPS26} reduces the translational volume packing of unit hyperdisks to ordering them along suitable stabbing lines, where the ordering can be computed using an approximation of the minimum-length Hamiltonian path using Christofides' algorithm, as the distances form a metric.
We replace this offline ordering step by our online TSP scheduling algorithm.

A unit hyperdisk in $\mathbb R^{d+1}$ is a unit-radius ball contained in a $d$-dimensional affine subspace. From now on, we assume that all hyperdisks in this section are unit hyperdisks.
Since rotations are not allowed, a hyperdisk is specified, up to translation, by the unit normal vector of its supporting hyperplane. We assume that the input hyperdisks are pairwise non-parallel.
For a finite set $\mathcal{D}$ of hyperdisks in $\mathbb{R}^{d+1}$, let $\OPT_\mathrm{AABB}(\mathcal D)$ be the minimum volume axis-aligned bounding box for a translational packing of $\mathcal{D}$.

Let $s \in \mathbb S^d$ be a unit vector in $\mathbb{R}^{d+1}$. Given a sequence of hyperdisks $D_1, \dots, D_m$, its \emph{stabbing length} (for short, \emph{span}) \emph{in direction $s$} is defined as follows: translate the hyperdisks such that their centers are on a line of direction vector $s$ in the given order so that consecutive hyperdisks touch but their relative interiors are disjoint; the stabbing length is the distance between the first and the last center. In general, for any translational packing of a set of hyperdisks $\mathcal{D}$, where the centers are on a line in direction $s$, the \emph{stabbing length} of the packing is the distance between the first and the last center.

The \textsc{DiskStabbing}$[\mathcal{D},s]$ problem for a finite set of hyperdisks $\mathcal D=\{D_1, \dots, D_m\}$ with respect to a direction $s\in \mathbb{S}^d$ is to find a linear order (i.e., a permutation) of the hyperdisks in $\mathcal{D}$ that minimizes the stabbing length in direction $s$.
Let $\OPT_\mathrm{stab}(\mathcal D,s)$ denote the optimum solution of \textsc{DiskStabbing}$[\mathcal{D},s]$.

We will use the following results (where $e_1,\ldots, e_{d+1}$ denote the standard basis vectors in $\mathbb{R}^{d+1}$). Remember, $d_s(D, D')$ is the separation needed along direction $s \in \mathbb S^d$ so that after translation hyperdisks $D$ and $D'$ do not overlap. 

\begin{lemma}[{\cite[Theorem 3]{AltCCPS26}}] \label{lemma:disks:metric}
    For any $s \in \mathbb S^d$, the function $d_s$ is a metric on the set of $d$-dimensional unit hyperdisks with normals not orthogonal to $s$ (and where parallel hyperdisks are considered equivalent). 
\end{lemma}

\begin{theorem}[{\cite[Theorem 6]{AltCCPS26}}]\label{lemma:disks:stabbing}
    Let $\mathcal D$ be a family of unit hyperdisks whose normals make an angle of at most $\phi_0 =\arccos(1/\sqrt{d+1})$ with the $x_{d+1}$-axis. 
    Then \[\OPT_\mathrm{stab}(\mathcal D,e_{d+1})\leq 2(d+1)^d \cdot \OPT_\mathrm{AABB}(\mathcal D).\]
\end{theorem}

\begin{lemma}[Online Stabbing]\label{lemma:disks:online:stabbing}
    For every direction $s$ and every online sequence $\mathcal D $ of at most $n$ unit hyperdisks whose normals are pairwise distinct and not orthogonal to $s$, there is an online algorithm that assigns a point $t(D) \in \mathbb{R}$ to each hyperdisk in $\mathcal{D}$ such that 
    \[|t(D) - t(D')| \geq d_s(D, D')\]
    for all $D, D'\in \mathcal{D}$, and whose final stabbing length in direction $s$ is $O(\OPT_\mathrm{stab}(\mathcal D,s) \log^2 n)$.
\end{lemma}
\begin{proof}
    By Lemma \ref{lemma:disks:metric}, $d_s$ is a metric. A Hamiltonian path in this metric is an ordering of hyperdisks and its length is the stabbing length in direction $s$. The lemma simply follows from our TSP scheduling algorithm of~\cref{thm:TSP:scheduling}.
\end{proof}
We now describe the online packing algorithm. When a hyperdisk with (unit) normal vector $u=(u_1, \dots, u_{d+1})$ arrives, we assign it to \emph{class} $i\in \{1,\ldots , d+1\}$ maximizing $|u_i|$. Then $|u_i| \geq 1 / \sqrt{d+1}$, so the angle between $u$ and the standard basis vector $e_i$ is at most $\phi_0$. For each class $i$, we use the stabbing algorithm from Lemma \ref{lemma:disks:online:stabbing} in direction $e_i$.

The algorithm works in two phases.
Let $i_1$ be the class of the first arriving hyperdisk. 
We stay in phase I as long as newly arriving hyperdisks belong to this class and their stabbing length is at most 1. 
If a hyperdisk of a different class arrives, or if the placement of the next hyperdisk from class $i_1$ would make the stabbing length exceed 1, we start a new phase. 

When we transition to phase~II, the span of all hyperdisks in phase I is at most 1 in direction $e_{i_1}$. We can enclose these hyperdisks in an axis-aligned hypercube of side length 3, because the centers of the hyperdisks in phase~I lie on an axis-aligned line segment of length at most 1, and all hyperdisks lie in the Euclidean 1-neighborhood of this line segment.
From now on, we use an approach similar to box packing in \Cref{sec:box-packing}: we always pack hyperdisks in axis-aligned hypercubes of side length 3 and stack consecutive hypercubes in the direction $e_{i_1}$. During phase~II, let $t_i(D)$ be the online stabbing coordinate assigned to a hyperdisk in class $i$.

For each class $i$, partition the stabbing line into intervals $[k, k+1)$, $k\in \mathbb Z$. Whenever the first hyperdisk of class $i$ whose assigned stabbing coordinate lies in such an interval arrives, i.e., $\lfloor t_i(D) \rfloor = k$, we open a new hypercube, and place the center of a hyperdisk with stabbing coordinate $t_i(D)$ at $(\tfrac{3}{2}, \dots, \tfrac{3}{2}) + (t_i(D) - k - \tfrac{1}{2})e_i$ inside that hypercube.

Then every center has distance at least 1 from every facet of the hypercube, so the hyperdisk is contained in the cube.

\mainDisks*

\begin{proof}
    Let $\mathcal{D}$ be a set of all hyperdisks, $\mathcal D_i$ the set of all hyperdisks of class $i$, and let $\OPT_\mathrm{STAB}(\mathcal D_i):=\OPT_\mathrm{stab}(\mathcal D_i,e_i)$ be the optimum stabbing length of $\mathcal D_i$ in direction $e_i$. Since the normal of every hyperdisk in $\mathcal D_i$ makes an angle at most $\phi_0$ with $e_i$, \cref{lemma:disks:stabbing} gives
    \[\OPT_\mathrm{STAB}(\mathcal D_i) \leq 2(d+1)^d \OPT_\mathrm{AABB}(\mathcal D_i). \]
    Denote the online stabbing length of class $i$ by $S_i$. By \cref{lemma:disks:online:stabbing}, we have $S_i = O_d(\OPT_\mathrm{AABB}(\mathcal D_i)\log^2 n)$.

    First, suppose the algorithm never switches to the second phase. Then all hyperdisks lie in the first class $i_1$ (that is, $\mathcal{D}_{i_1}=\mathcal{D}$), and are placed on a stabbing line of span at most $O_d(\OPT_\mathrm{AABB}(\mathcal D_{i_1})\log^2 n)$. Let $E_j$ be the maximum extent in coordinate direction $e_j$ (when the centers are placed at the origin). The axis-aligned bounding box enclosing our disk packing has volume at most
    \[ \left(E_{i_1} + O_d(\OPT_\mathrm{AABB}(\mathcal D) \log^2 n)\right)\cdot \prod_{j\neq i_1}E_j.\]
We can bound $\prod_{j=1}^{d+1}E_j$ simply by $\OPT_\mathrm{AABB}(\mathcal D)$ and since $E_j \leq 2$ for all $j$, we have $\prod_{j \neq i_1} E_j \leq 2^d$ and thus the enclosing volume is just
    \[O_d(\OPT_\mathrm{AABB}(\mathcal D)\log^2 n).\]

    Now suppose the algorithm switches to the second phase at some point. The disks from the first phase are contained in one hypercube of side length 3. 
    The relative interiors of all hyperdisks are disjoint which follows as the hypercubes themselves are disjoint and within each hypercube this follows from the disjointness of the hyperdisks stabbing algorithm.
    The number of second-phase hypercubes opened for any class $i$ is bounded by $\lceil S_i \rceil +2$. Hence the total number of hypercubes is 
    \[1+\sum_{i=1}^{d+1}S_i 
    = O_d\left(\log^2 n \sum_{i=1}^{d+1}\OPT_\mathrm{AABB} (\mathcal D_i)\right) 
    = O_d\left(\log^2 n \OPT_\mathrm{AABB} (\mathcal D)\right).\]
    Since each hypercube has volume $3^{d+1}$, we get an asymptotic competitive ratio of $O_d(\log^2 n)$.

    We can also eliminate the additive constant, by proving that we have $\OPT_{\mathrm{AABB}}(\mathcal D) \cdot \log^2 n = \Omega_d(1)$. If the switch of phases occurred because the next hyperdisk would have increased the span beyond 1, then for the subset of hyperdisks seen up to the switch, we have 
    \[1 < S_{i_1} \leq O_d(\log^2 n \cdot \OPT_\mathrm{AABB}).\]
    
    If the switch occurred because a second class appeared, then at least two classes are non-empty. For every direction $e_r$, choose a non-empty class $i\neq r$; the unit normal vector of a hyperdisk in class $i$ satisfies $|u_r|\leq \tfrac{1}{\sqrt 2}$, and therefore its projection in direction $e_r$ has constant length, thus $\OPT_\mathrm{AABB} \geq \prod_j E_j = \Omega(1)$, again implying $\OPT_{\mathrm{AABB}}(\mathcal D) \cdot \log^2 n = \Omega_d(1)$.
\end{proof}

The following theorem shows that it is impossible to extend the approach above to an $O(\polylog n)$-competitive algorithm for disks of arbitrary radii. In particular, the theorem gives the lower bound $\Omega(n^{1/6})$ for 2-dimensional disks in $\mathbb R^3$.

\lowerDisks*

\begin{proof}
From~\cite{AltCCPS26}, we know that there is a family $\mathcal D$ of $n-1$ unit hyperdisks in $\mathbb R^{d+1}$ such that any convex container where we can pack $\mathcal D$ under translations
has volume $\Omega(n^{\frac{d-1}{d}})$.
    The adaptive adversary first presents the $n-1$ unit hyperdisks from $\mathcal D$ and then one large hyperdisk $D^*$ specified below.

    Let $L_1, \dots, L_{d+1}$ be the length of the algorithm's current axis-aligned bounding box after the $n-1$ arrivals.
    We thus have 
    \[\prod_{j=1}^{d+1}L_j \geq \Omega\left(n^\frac{d-1}{d}\right).\]
    Hence there exists some coordinate direction $e_i$ such that \[L_i \geq \Omega\left(n^\frac{d-1}{d(d+1)}\right).\]
    The adversary now presents one final hyperdisk $D^*$ of radius $n^{1/d}$ whose normal vector is $e_i$. Thus it has width 0 in direction $e_i$ and $2n^{1/d}$ in every coordinate direction other than $e_i$.

    The online algorithm cannot move the first $n$ hyperdisks. Therefore its final bounding box still has side length at least $L_i$ in direction $e_i$. It must also have side length $\Omega(n^{1/d})$ in each of the remaining $d$ directions because of $D^*$. Thus the final volume of the packing produced by the online algorithm is at least 
    \[\Omega\left({(n^{1/d})}^d \cdot L_i\right) = \Omega\left(n^{1+\frac{d-1}{d(d+1)}}\right).\]
    
    On the other hand, the offline optimum for the whole instance is $O(n)$:
    The large hyperdisk fits in a box of dimensions $O(n^{1/d}\times \dots \times 1 \times \dots \times n^{1/d} )$. We can place the $n-1$ unit hyperdisks in a box of the same dimensions. Since every unit hyperdisk is contained in the unit ball around its center, we can place their centers on a $d$-dimensional grid of side length $O(n^{1/d})$ and constant spacing, which gives a feasible packing contained in an axis-aligned box with $d$ side lengths of $O(n^{1/d})$ and one side length $O(1)$, hence the theorem follows.
\end{proof}

\section{Online TSP Ordering}\label{sec:TSP}
We now prove~\cref{thm:TSP}. The proof is conceptually similar to the proof of~\cref{thm:TSP:scheduling}. It employs the elementary trees by Azar, Panigrahi, and Vardi~\cite{azar2026nearly} having to overcome the same obstacles. Let us define the problem formally. 
\begin{definition}[Online TSP Ordering]
Let $(M,d)$ be a metric space and $\eps\geq 0$. Points $x_1,\ldots,x_n\in M$ arrive online,
and upon arrival each $x_i$ must be assigned to an empty cell of an array $A$ consisting of $(1+\eps)n$ cells. After all $n$ points have been inserted, let $z_1,\dots,z_n$ be the non-empty entries of the array going from left to right, and define the cost as $\sum_{i=1}^{n-1}d(z_i,z_{i+1})$.
\end{definition}

We denote this problem \OnlineTSPordering$[n,\eps,M,d]$.
We first consider the case where $\OPT$ is known to the algorithm.

\begin{lemma}\label{lemma:known:OPT}
Let $\eps\in (0, 1]$. There exists an algorithm for
\OnlineTSPordering$[n,\eps,M,d]$ with competitive ratio
$O((\log^2 n)/\eps)$ if $\OPT$ is given as part of the input.
\end{lemma}

\begin{proof}
By normalizing, we may assume that $\OPT=1$ (the treatment of $\OPT=0$ is trivial).  We may also clearly assume that $n$ is sufficiently large. Finally, we may assume that $\eps>1/\sqrt{n}$, since otherwise, we can apply the $O(\sqrt{n})$ competitive algorithm by Bertram~\cite{DBLP:conf/esa/Bertram25} which works even with $\eps=0$.
We again use the \emph{elementary tree} data structure~\cite{azar2026nearly} but with the same new labeling scheme and modification of how elements are inserted into this data structure as in the proof of~\cref{thm:online:tsp:scheduling:known:opt}. Recall the notation from that section: $H$ is the height of an elementary tree with $2^H$ leaves. The height of the leaves is zero, so that the height of the root is $H$. 
We can use an elementary tree $T$ of height $H$ to insert elements into an array of size $2^H$. If a DFS order of the leaves of $T$ are $v_1,\dots,v_{2^H}$, then we can let elements inserted at $v_i$ be placed in the $i$'th cell of the array. We make the following claim.

\begin{claim}\label{claim:cost}
The cost of any solution produced by inserting elements into an array of size $2^H$ using an elementary tree of height $H$ is at most $2H $.
\end{claim}

\begin{proof}[Proof of Claim.]
    Let $B$ be the array. Let $c_1,\dots c_k$ denote the non-empty cells of $B$ going from left to right, and let $z_1,\dots,z_k\in M$ denote their content, so that the cost of the solution is $\sum_{i=1}^{k-1}d(z_{i},z_{i+1})$. Let $v_1,\dots, v_k$ denote the leaves of $T$ corresponding to cells $c_1,\dots,c_k$. For $1\leq i\leq k-1$, let $u_i$ denote the least common ancestor of $v_i$ and $v_{i+1}$. It is easy to check that the mapping $i\mapsto u_i$ is injective for $1\leq i\leq k-1$. By the same ball-containment argument used in \cref{sec:TSP-scheduling}, if for some $i$, $u_i$ is at height $h$, then $d(z_{i},z_{i+1})\leq 2^{h-H+1}$.
    Since the number of nodes at height $h$ is $2^{H-h}$ and the mapping $i\mapsto u_i$ is injective, we conclude that 
    \[
    \sum_{i=1}^{k-1}d(z_{i},z_{i+1})
    \leq \sum_{h=1}^H 2^{H-h}\cdot 2^{h-H+1}
    \leq 2H .
    \qedhere
    \]    
\end{proof}

The problem with using a single elementary tree of size $n(1+\eps)$ is that in the worst-case, it will only have space for $O(\log n)$ elements. To resolve this, we define $\ell=n(1+\eps)/2^H$ and partition the array $A$ into subarrays $B_1,\dots, B_{\ell}$ each of size $2^H$, where we will fix $H$ later.
We build an elementary tree $T_i$ on top of each $B_i$ and when a new element $x$ arrives, it is inserted along a feasible path of maximum labeled depth among all paths in all elementary trees. If the maximum labeled depth is $0$, that is, if we start a new tree, we choose the leftmost tree whose root is still unlabeled. All other ties are broken arbitrarily. With enough trees, it turns out that we can insert all $n$ elements. To prove this, we use~\cref{claim:sched:partial:separation}, which we restate here for convenience.

\begin{claim}\label{claim:labels}
Let $v$ and $v'$ be distinct partial nodes of height $h$, possibly in different elementary trees. If $y$ and $y'$ are their labels, then \[ d(y,y')> 2^{h-H-1} . \]
\end{claim}

Call a node \emph{partial} if it has exactly one labeled child. By Claim~\ref{claim:labels}, for every fixed height $h$, the labels of partial nodes at height $h$ are $2^{h-H-1}$-separated. Since $\OPT=1$, any set of input points that is $\delta$-separated has size at most $1/\delta+1$, by an identical argument to the one given in the proof of~\cref{thm:online:tsp:scheduling:known:opt}. Thus the number of partial nodes at height $h$ is at most $2^{H-h+1}+1$.

Every empty cell in a non-empty elementary tree lies in the subtree rooted at the unlabeled child of a unique partial node. A partial node at height $h$ contributes $2^{h-1}$ such empty cells. Hence the total number of empty cells in non-empty elementary trees is at most
\[
    \sum_{h=1}^H (2^{H-h+1}+1)2^{h-1}
    \leq (H+1)2^H .
\]
We choose $H$ maximal such that $(H+1)2^H\leq \eps n$. Such a $H$ exists by the assumption that $\eps>1/\sqrt{n}$ and that $n$ is sufficiently large. Then, $2^H=\Omega(\eps n/\log n)$, and so $\ell=O(\log n/\eps)$.

It remains to see that the algorithm never fails before all $n$ elements have been inserted. If it failed before inserting some element, then no empty elementary tree could remain, since starting a new tree is always feasible. Thus all remaining empty cells would lie in non-empty elementary trees. But before $n$ insertions have been made, the array contains more than $\eps n$ empty cells, contradicting the bound $(H+1)2^H\leq \eps n$ above.

Finally, by Claim~\ref{claim:cost}, the cost inside each elementary tree is at most $2H$. Since $\OPT=1$, the diameter of the input point set is at most $1$, and hence the cost of moving between consecutive non-empty elementary trees is at most $1$ per transition. Therefore the total cost is at most
\[
    \ell(2H+1)=O\left(\frac{\log^2 n}{\eps}\right),
\]
as desired.
\end{proof}

\subsection{Unknown OPT}

To obtain~\cref{thm:TSP}, we need to consider the setting when $\OPT$ is not part of the input. The technique is conceptually similar to the analogous proof in~\cite{azar2026nearly}, but we include the argument for completeness.

\begin{lemma}\label{lemma:TSP:unknown:OPT}
Let $\eps\in (0, 1]$. There exists an algorithm for
\OnlineTSPordering$[n,\eps,M,d]$ with competitive ratio $O((\log^3 n)/\eps)$ if $\OPT$ is unknown to the algorithm.
\end{lemma}
\begin{proof}
We may again assume that $\eps>1/\sqrt{n}$ and that $n$ is sufficiently large. For a prefix $x_1,\dots,x_t$, let $\OPT_t$ denote the cost of an optimal traveling salesperson path through these points. The algorithm will not compute $\OPT_t$. Instead, let $a_t$ be the cost of a path returned by a fixed polynomial-time $2$-approximation algorithm for metric TSP path on $x_1,\dots,x_t$, for instance the standard approximation arising from an MST. Thus
\[
    \OPT_t\leq a_t\leq 2\OPT_t .
\]
Define
\[
    D_t=\max_{s\leq t} a_s .
\]
Since $\OPT_t$ is non-decreasing, we have
\[
    \OPT_t\leq D_t\leq 2\OPT_t .
\]

If $D_t=0$, then the first $t$ points are at pairwise distance zero, and we insert them consecutively at the first $t$ cells of the array. Suppose now that $D_t>0$ for the first time. We start the first epoch with guess
\[
    \Delta_1=2D_t .
\]
More generally, during an epoch with guess $\Delta_i$, before inserting a new element $x_t$ we compute $D_t$. If $D_t>\Delta_i$, the current epoch ends, and a new epoch starts with guess
\[
    \Delta_{i+1}=2D_t .
\]
Thus $\Delta_{i+1}>2\Delta_i$, while every $\Delta_i$ is at most $2D_n\leq 4\OPT$. Hence
\[
    \sum_i \Delta_i=O(\OPT).
\]

Fix an epoch with guess $\Delta_i$. Every set of elements inserted during this epoch has optimal TSP-path cost at most $\Delta_i$, since each such set is contained in a prefix $x_1,\dots,x_t$ with $D_t\leq \Delta_i$ at the time of insertion.

Let $\eta=\eps/10$, and let
\[
    N_0=\left\lceil \frac{100\log n}{\eta}\right\rceil .
\]
The first $N_0$ elements of an epoch are inserted consecutively. The remaining elements of the epoch are handled in phases of capacities
\[
    N_0,2N_0,4N_0,\dots .
\]
For a phase of capacity $N$, choose $H$ maximal such that
\[
    (H+1)2^H\leq \eta N .
\]
The phase opens elementary trees of height $H$ from left to right, opening a new tree only when the next element has no feasible path in any tree already open. Within the currently open trees of the phase, a path is feasible for $x$ if either it contains no labeled node, or, letting $v_h$ be its deepest labeled node and $y$ the label of $v_h$, we have
\[
    d(x,y)\leq \Delta_i 2^{h-H-1}.
\]
The element is inserted along a feasible path of maximum labeled depth among all currently open trees of the phase. If no such path exists, a new elementary tree is opened immediately to the right of the previous trees of the phase, and $x$ is inserted there. All other ties are broken arbitrarily. A phase ends after receiving $N$ elements, unless the epoch ends first.

We first bound the number of cells used. Consider a phase of capacity $N$ in epoch $i$, and suppose it receives $q\leq N$ elements. The proof of Claim~\ref{claim:labels}, with all distances scaled by $\Delta_i$, shows that at every height $h$, the labels of partial nodes of height $h$ in the elementary trees of this phase are $\Delta_i2^{h-H-1}$-separated. Since the elements of the phase have optimal TSP-path cost at most $\Delta_i$, the number of partial nodes of height $h$ is at most
\[
    2^{H-h+1}+1 .
\]
A partial node of height $h$ accounts for at most $2^{h-1}$ empty cells. Therefore the total number of empty cells in non-empty elementary trees of the phase is at most
\[
    \sum_{h=1}^H (2^{H-h+1}+1)2^{h-1}
    \leq (H+1)2^H
    \leq \eta N .
\]
Since trees are opened only when they become non-empty, the number of cells allocated to this phase is at most
\[
    q+\eta N .
\]

Now consider an epoch containing $n_i$ elements. The initial consecutive part uses exactly one cell per inserted element. The capacities of the phases started in the epoch form a geometric sequence. If the last started phase has capacity $L$, then the sum of the previous phase capacities is $L-N_0$, and hence $L<n_i$. Thus the sum of all started phase capacities is less than $2n_i$. Hence the epoch allocates at most
\[
    n_i+2\eta n_i
\]
cells. Summing over all epochs, the total number of allocated cells is at most
\[
    (1+2\eta)n\leq (1+\eps)n .
\]

It remains to bound the cost. Consider again a phase of capacity $N$ and height $H$ in epoch $i$. The number of non-empty elementary trees in the phase is at most
\[
    \frac{(1+\eta)N}{2^H}.
\]
By maximality of $H$,
\[
    \eta N < (H+2)2^{H+1},
\]
and therefore the number of non-empty elementary trees is $O(H/\eta)$. By the scaled version of Claim~\ref{claim:cost}, each elementary tree contributes cost at most $2\Delta_i H$. The cost between consecutive non-empty elementary trees in the same phase is at most $\Delta_i$, since the elements of the phase have diameter at most $\Delta_i$. Hence the cost of one phase is
\[
    O\left(\frac{\Delta_i H^2}{\eta}\right)
    =
    O\left(\frac{\Delta_i\log^2 n}{\eps}\right).
\]

There are $O(\log n)$ phases in one epoch. The initial consecutive part contributes at most $N_0\Delta_i=O(\Delta_i\log n/\eps)$, and the transitions between consecutive parts of the same epoch contribute $O(\Delta_i\log n)$. Therefore the total cost of epoch $i$ is
\[
    O\left(\frac{\Delta_i\log^3 n}{\eps}\right).
\]
The cost between two consecutive epochs is at most the guess of the later epoch, and these costs are absorbed by $\sum_i O(\Delta_i)$. Since $\sum_i\Delta_i=O(\OPT)$, the total cost is
\[
    O\left(\frac{\OPT\log^3 n}{\eps}\right).
    \qedhere
\]
\end{proof}

Now \cref{thm:TSP} follows directly from~\cref{lemma:known:OPT} and~\cref{lemma:TSP:unknown:OPT}.

\mainTSP*

\section{Open Problems}
\label{sec:con}
We list here some open problems. We find \OnlineTSPscheduling~interesting in its own right, due to its applicability to geometric packing problems but also to the scheduling problems discussed in the introduction. The only lower bound we are aware of for the competitive ratio is $\Omega\left(\sqrt\frac{\log n}{\log \log n}\right)$ which is inherited from the lower bound for online strip packing in~\cite{aamand2023online}. It would be interesting to narrow the gap to the $O(\log^2 n)$ upper bound. In particular, we wonder if it is possible to obtain a better lower bound using the fact that we have more flexibility to choose the metric. 

For online translational strip packing for convex pieces, we show that the optimal competitive ratio is polylogarithmic in $n$. Can the competitive ratio be further reduced, or can the lower bound be improved? A similar question can be asked for the other packing variants. 

Finally, we know by \cref{thm:arbitrary:radius:hyperdisk:lower} that it is impossible to obtain polylogarithmic competitive ratios for online translational packing of general convex pieces in three dimensions.
In fact, whereas in the two-dimensional version of the problem, we can relate the cost of an optimal offline solution to a TSP cost, in three dimensions, we do not even have a grasp of the offline optimal cost.
Is there an approximation algorithm with non-trivial approximation guarantees for, e.g., volume packing or strip packing of convex polyhedra?
What can be said about online algorithms?

\subsection*{Acknowledgements}
This work was initiated at the Dagstuhl Seminar 25372: Precision in Geometric Algorithms.
We thank the organizers and participants for a fruitful atmosphere.

\bibliography{ref}

@article{Allahverdi1999,
  title={A review of scheduling research involving setup considerations},
  author={Allahverdi, Ali and Gupta, Jatinder N.D.\ and Aldowaisan, Tariq},
  journal={Omega},
  volume={27},
  number={2},
  pages={219--239},
  year={1999},
  doi={10.1016/S0305-0483(98)00042-5}
}

@article{allahverdi2008survey,
  title={A survey of scheduling problems with setup times or costs},
  author={Allahverdi, Ali and Ng, Chi To and Cheng, T.C. Edwin and Kovalyov, Mikhail Y.},
  journal={European Journal of Operational Research},
  volume={187},
  number={3},
  pages={985--1032},
  year={2008},
  doi={10.1016/j.ejor.2006.06.060}
}

@article{LinYing2022,
  title={Single machine scheduling problems with sequence-dependent setup times and precedence delays},
  author={Lin, Shih-Wei and Ying, Kuo-Ching},
  journal={Scientific Reports},
  volume={12},
  number={1},
  pages={9430},
  year={2022},
  publisher={Nature Publishing Group UK London},
  doi={10.1038/s41598-022-13278-y}
}

@misc{LeibHellerKuehn2025,
  title={Computing an optimal single machine schedule with sequence dependent setup times using shortest path computations},
  author={Leib, Dominik and Heller, Till and K{\"u}hn, Raphael},
  howpublished={Preprint},
  eprint={2507.20611},
  year={2025}
}

@article{AltCCPS26,
  author       = {Helmut Alt and
                  Sergio Cabello and
                  Otfried Cheong and
                  Ji{-}won Park and
                  Nadja Seiferth},
  title        = {Packing $d$-dimensional balls into a $d+1$-dimensional
                  container},
  journal      = {Comput. Geom.},
  volume       = {132},
  pages        = {102219},
  year         = {2026},
  doi          = {10.1016/J.COMGEO.2025.102219},
}

@article{MiyazawaPSSW16,
  author       = {Fl{\'{a}}vio Keidi Miyazawa and
                  Lehilton L. C. Pedrosa and
                  Rafael Crivellari Saliba Schouery and
                  Maxim Sviridenko and
                  Yoshiko Wakabayashi},
  title        = {Polynomial-Time Approximation Schemes for Circle and Other Packing Problems},
  journal      = {Algorithmica},
  volume       = {76},
  number       = {2},
  pages        = {536--568},
  year         = {2016},
  _url          = {https://doi.org/10.1007/s00453-015-0052-4},
  doi          = {10.1007/S00453-015-0052-4}
}

@inproceedings{KarKW26,
  author       = {Debajyoti Kar and
                  Arindam Khan and
                  Andreas Wiese},
  title        = {Approximation Schemes and Structural Barriers for the Two-Dimensional
                  Knapsack Problem with Rotations},
  booktitle    = {Proc.\ 58th {ACM} Symposium on Theory of Computing (STOC)},
  pages        = {1145--1156},
  _publisher    = {{ACM}},
  year         = {2026},
  _url          = {https://doi.org/10.1145/3798129.3800826},
  doi          = {10.1145/3798129.3800826},
  biburl       = {https://dblp.org/rec/conf/stoc/KarKW26.bib},
  bibsource    = {dblp computer science bibliography, https://dblp.org}
}

@inproceedings{AcharyaBG0MW24,
  author       = {Pritam Acharya and
                  Sujoy Bhore and
                  Aaryan Gupta and
                  Arindam Khan and
                  Bratin Mondal and
                  Andreas Wiese},
  title        = {Approximation Schemes for Geometric Knapsack for Packing Spheres and
                  Fat Objects},
  booktitle    = {Proc.\ 51st International Colloquium on Automata, Languages, and Programming                   (ICALP)},
  series       = {LIPIcs},
  volume       = {297},
  pages        = {8:1--8:20},
  publisher    = {Schloss Dagstuhl},
  year         = {2024},
  _url          = {https://doi.org/10.4230/LIPIcs.ICALP.2024.8},
  doi          = {10.4230/LIPICS.ICALP.2024.8}
}

@inproceedings{BansalK14,
  author       = {Nikhil Bansal and
                  Arindam Khan},
  title        = {Improved Approximation Algorithm for Two-Dimensional Bin Packing},
  booktitle    = {Proc.\ 25th {ACM-SIAM} Symposium on Discrete Algorithms (SODA)},
  pages        = {13--25},
  _publisher    = {{SIAM}},
  year         = {2014},
  doi          = {10.1137/1.9781611973402.2}
}

@article{GalvezGIHKW21,
  author       = {Waldo G{\'{a}}lvez and
                  Fabrizio Grandoni and
                  Salvatore Ingala and
                  Sandy Heydrich and
                  Arindam Khan and
                  Andreas Wiese},
  title        = {Approximating Geometric Knapsack via {$L$}-packings},
  journal      = {{ACM} Trans. Algorithms},
  volume       = {17},
  number       = {4},
  pages        = {33:1--33:67},
  year         = {2021},
  doi          = {10.1145/3473713}
}

@inproceedings{Kar0R25,
  author       = {Debajyoti Kar and
                  Arindam Khan and
                  Malin Rau},
  title        = {Improved Approximation Algorithms for Three-Dimensional Bin Packing},
  booktitle    = {Proc.\ 52nd International Colloquium on Automata, Languages, and Programming (ICALP)},
  series       = {LIPIcs},
  volume       = {334},
  pages        = {104:1--104:20},
  publisher    = {Schloss Dagstuhl},
  year         = {2025},
  doi          = {10.4230/LIPICS.ICALP.2025.104}
}

@inproceedings{JansenK0ST25,
  author       = {Klaus Jansen and
                  Debajyoti Kar and
                  Arindam Khan and
                  K. V. N. Sreenivas and
                  Malte Tutas},
  title        = {Improved Approximation Algorithms for Three-Dimensional Knapsack},
  booktitle    = {Proc.\ 41st Symposium on Computational Geometry (SoCG)},
  series       = {LIPIcs},
  volume       = {332},
  pages        = {60:1--60:18},
  publisher    = {Schloss Dagstuhl},
  year         = {2025},
  doi          = {10.4230/LIPICS.SOCG.2025.60}
}

@inproceedings{Jansen0LS22,
  author       = {Klaus Jansen and
                  Arindam Khan and
                  Marvin Lira and
                  K. V. N. Sreenivas},
  title        = {A {PTAS} for Packing Hypercubes into a Knapsack},
  booktitle    = {Proc.\ 49th International Colloquium on Automata, Languages, and Programming (ICALP)},
  series       = {LIPIcs},
  volume       = {229},
  pages        = {78:1--78:20},
  publisher    = {Schloss Dagstuhl},
  year         = {2022},
  doi          = {10.4230/LIPICS.ICALP.2022.78}
}

@inproceedings{aamand2023online,
  title={Online sorting and translational packing of convex polygons},
  author={Aamand, Anders and Abrahamsen, Mikkel and Beretta, Lorenzo and Kleist, Linda},
  booktitle={Proc.\ 34th ACM-SIAM Symposium on Discrete Algorithms (SODA)},
  pages={1806--1833},
  year={2023},
  doi= {10.1137/1.9781611977554.CH69}}

@inproceedings{azar2026nearly,
  title={Nearly Tight Bounds for the Online Sorting Problem},
  author={Azar, Yossi and Panigrahi, Debmalya and Vardi, Or},
  booktitle={Proc.\ 37th ACM-SIAM Symposium on Discrete Algorithms (SODA)},
  pages={6642--6658},
  year={2026},
  doi={10.1137/1.9781611978971.237}
}

@inproceedings{azar2026beyond,
  title={Online Metric {TSP}: {B}eyond the $\sqrt{n}$ Barrier},
  author={Azar, Yossi and Panigrahi, Debmalya and Vardi, Or},
  booktitle={Proc.\ 53rd International Colloquium on Automata, Languages, and Programming (ICALP)},
  pages={18:1--18:18},
  year={2026},
  publisher={Schloss Dagstuhl},
  series={LIPIcs},
  doi={ 10.4230/LIPIcs.ICALP.2026.18}
}

@InProceedings{OrthoOnlineSoCG,
  author =	{Gerlach, Tim and Hennies, Benjamin and Kleist, Linda},
  title =	{{Online Packing of Orthogonal Polygons}},
  booktitle =	{Symposium on Computational Geometry (SoCG 2026)},
  pages =	{52:1--52:17},
  year =	{2026},
  volume =	{367},
  doi =		{10.4230/LIPIcs.SoCG.2026.52}
}

@incollection{epstein2018multidimensional,
  title={Multidimensional packing problems},
  author={Epstein, Leah and van Stee, Rob},
  booktitle={Handbook of Approximation Algorithms and Metaheuristics},
  edition={2nd},
  pages={553--570},
  year={2018},
  editor={Teofilo F. Gonzalez},
  publisher={Chapman and Hall/CRC},
  doi={10.1201/9781351236423-31}
}

@article{CHRISTENSEN201763,
title = "Approximation and online algorithms for multidimensional bin packing: A survey",
journal = "Computer Science Review",
volume = "24",
pages = "63--79",
year = "2017",
issn = "1574-0137",
doi = "10.1016/j.cosrev.2016.12.001",
author = "Henrik I. Christensen and Arindam Khan and Sebastian Pokutta and Prasad Tetali"
}

@article{DBLP:journals/sigact/Stee15,
  author    = {Rob van Stee},
  title     = {{SIGACT} News Online Algorithms Column 26: {B}in packing in multiple
               dimensions},
  journal   = {{SIGACT} News},
  volume    = {46},
  number    = {2},
  pages     = {105--112},
  year      = {2015},
  doi       = {10.1145/2789149.2789167},
  bibsource = {dblp computer science bibliography, https://dblp.org}
}

@InProceedings{han2007strip,
author="Han, Xin
and Iwama, Kazuo
and Ye, Deshi
and Zhang, Guochuan",
title="Strip Packing vs.\ Bin Packing",
booktitle="Proc.\ 3rd Algorithmic Aspects in Information and Management (AAIM)",
year="2007",
pages="358--367",
volume={4508},
series={LNCS},
publisher={Springer},
doi={10.1007/978-3-540-72870-2_34}
}

@article{fekete2017online,
  author    = {S{\'{a}}ndor P. Fekete and
               Hella{-}Franziska Hoffmann},
  title     = {Online Square-into-Square Packing},
  journal   = {Algorithmica},
  volume    = {77},
  number    = {3},
  pages     = {867--901},
  year      = {2017},
  doi       = {10.1007/s00453-016-0114-2}
}

@article{alt_convexOffline_JoCG,
	title={Approximating minimum-area rectangular and convex containers for packing convex polygons},
	author={Alt, Helmut and de Berg, Mark and Knauer, Christian},
	journal={Journal of Computational Geometry},
	pages={1--10},
	year={2017},
	volume={8},
	number={1},
	doi={10.20382/jocg.v8i1a1}
}

@article{alt_convexOffline_corr,
	title={Corrigendum to: Approximating minimum-area rectangular and convex containers for packing convex polygons},
	author={Alt, Helmut and de Berg, Mark and Knauer, Christian},
	journal={Journal of Computational Geometry},
	pages={653--655},
	year={2020},
	volume={11},
	number={1},
	doi={10.20382/jocg.v11i1a26}
}

@inproceedings{Milenkovic96,
  author    = {Victor Milenkovic},
  title     = {Translational Polygon Containment and Minimal Enclosure using Linear
               Programming Based Restriction},
  booktitle = {Proc.\ 28th {ACM} Symposium on the Theory
               of Computing (STOC)},
  pages     = {109--118},
  year      = {1996},
  doi       = {10.1145/237814.237840},
  bibsource = {dblp computer science bibliography, https://dblp.org}
}

@article{doi:10.1111/j.1475-3995.1999.tb00171.x,
author = {Milenkovic, Victor J. and Daniels, Karen},
title = {Translational polygon containment and minimal enclosure using mathematical programming},
journal = {International Transactions in Operational Research},
volume = {6},
number = {5},
pages = {525--554},
doi = {10.1111/j.1475-3995.1999.tb00171.x},
year = {1999}
}

@article{MILENKOVIC19993,
title = "Rotational polygon containment and minimum enclosure using only robust {2D} constructions",
journal = "Computational Geometry",
volume = "13",
number = "1",
pages = "3--19",
year = "1999",
issn = "0925-7721",
doi = "10.1016/S0925-7721(99)00006-1",
author = "Victor J. Milenkovic",
}

@article{DBLP:journals/algorithmica/Milenkovic97,
  author       = {Victor Milenkovic},
  title        = {Multiple Translational Containment. {P}art {II:~E}xact Algorithms},
  journal      = {Algorithmica},
  volume       = {19},
  number       = {1/2},
  pages        = {183--218},
  year         = {1997},
  url          = {https://doi.org/10.1007/PL00014416},
  doi          = {10.1007/PL00014416},
  bibsource    = {dblp computer science bibliography, https://dblp.org}
}

@inproceedings{AbrahamsenBeretta20,
  author =	{Abrahamsen, Mikkel and Beretta, Lorenzo},
  title =	{Online Packing to Minimize Area or Perimeter},
  booktitle =	{Proc.\ 37th Symposium on Computational Geometry (SoCG)},
  pages =	{6:1--6:15},
  year =	{2021},
  series={LIPIcs}, 
  volume=189, 
  publisher={Schloss Dagstuhl},
  doi =		{10.4230/LIPIcs.SoCG.2021.6}
}

@article{Csirik1997shelf,
  author    = {J{\'{a}}nos Csirik and
               Gerhard J. Woeginger},
  title     = {Shelf Algorithms for On-Line Strip Packing},
  journal   = {Information Processing Letters},
  volume    = {63},
  number    = {4},
  pages     = {171--175},
  year      = {1997},
  doi       = {10.1016/S0020-0190(97)00120-8}
}

@article{baker1983shelf,
  title={Shelf algorithms for two-dimensional packing problems},
  author={Baker, Brenda S. and Schwarz, Jerald S.},
  journal={SIAM Journal on Computing},
  volume={12},
  number={3},
  pages={508--525},
  year={1983},
  publisher={SIAM},
  doi = {10.1137/0212033}
}

@article{YeOnlineStrip,
title={A note on online strip packing},
author={Ye, Deshi and Han, Xin and Zhang, Guochuan},
journal={Journal of Combinatorial Optimization},
volume={17},
number={4},
pages={417--423},
year={2009},
doi={10.1007/s10878-007-9125-x},
publisher={Springer}
}

@inproceedings{DBLP:conf/icalp/MerinoW20,
  author    = {Arturo I. Merino and
               Andreas Wiese},
  title     = {On the Two-Dimensional Knapsack Problem for Convex Polygons},
  booktitle = {Proc.\ 47th International Colloquium on Automata, Languages, and Programming ({ICALP})},
  series    = {LIPIcs},
  pages     = {84:1--84:16},
  year      = {2020},
  doi       = {10.4230/LIPIcs.ICALP.2020.84},
  bibsource = {dblp computer science bibliography, https://dblp.org}
}

@article{DBLP:journals/theoretics/AbrahamsenMS24,
  author       = {Mikkel Abrahamsen and
                  Tillmann Miltzow and
                  Nadja Seiferth},
  title        = {Framework for $\exists\mathbb{R}$-Completeness of Two-Dimensional
                  Packing Problems},
  journal      = {TheoretiCS},
  volume       = {3},
  eid          = {11},
  year         = {2024},
  _url          = {https://doi.org/10.46298/theoretics.24.11},
  doi          = {10.46298/THEORETICS.24.11},
  timestamp    = {Wed, 11 Mar 2026 13:25:33 +0100},
  biburl       = {https://dblp.org/rec/journals/theoretics/AbrahamsenMS24.bib},
  bibsource    = {dblp computer science bibliography, https://dblp.org}
}

@inproceedings{DBLP:conf/esa/KurpiszS23,
  author       = {Adam Kurpisz and
                  Silvan Suter},
  title        = {Improved Approximations for Translational Packing of Convex Polygons},
  booktitle    = {Proc.\ 31st European Symposium on Algorithms ({ESA})},
  pages        = {76:1--76:14},
  year         = {2023},
  doi          = {10.4230/LIPICS.ESA.2023.76},
  bibsource    = {dblp computer science bibliography, https://dblp.org}
}

@article{coffman1980performance,
  title={Performance bounds for level-oriented two-dimensional packing algorithms},
  author={Coffman, Jr., Edward G. and Garey, Michael R. and Johnson, David S. and Tarjan, Robert Endre},
  journal={SIAM Journal on Computing},
  volume={9},
  number={4},
  pages={808--826},
  year={1980},
  doi= {10.1137/0209062}
}

@article{epstein2019lower,
	title={A lower bound for online rectangle packing},
	author={Epstein, Leah},
	journal={Journal of Combinatorial Optimization},
	volume={38},
	pages={846--866},
	year={2019},
	doi={10.1007/s10878-019-00423-z}
}

@article{balogh2019lower,
	title={Lower bounds for several online variants of bin packing},
	author={Balogh, J{\'a}nos and B{\'e}k{\'e}si, J{\'o}zsef and D{\'o}sa, Gy{\"o}rgy and Epstein, Leah and Levin, Asaf},
	journal=tocs,
	volume={63},
	pages={1757--1780},
	year={2019},
	doi={10.1007/s00224-019-09915-1}
}

@inproceedings{DBLP:conf/esa/AbrahamsenB0K024,
  author       = {Mikkel Abrahamsen and
                  Ioana O. Bercea and
                  Lorenzo Beretta and
                  Jonas Klausen and
                  L{\'{a}}szl{\'{o}} Kozma},
  _editor       = {Timothy M. Chan and  Johannes Fischer and John Iacono and Grzegorz Herman},
  title        = {Online Sorting and Online {TSP}: {R}andomized, Stochastic, and High-Dimensional},
  booktitle    = {Proc.\ 32nd European Symposium on Algorithms (ESA)},
  series       = {LIPIcs},
  volume       = {308},
  pages        = {5:1--5:15},
  publisher    = {Schloss Dagstuhl},
  year         = {2024},
  _url          = {https://doi.org/10.4230/LIPIcs.ESA.2024.5},
  doi          = {10.4230/LIPICS.ESA.2024.5},
  timestamp    = {Fri, 21 Nov 2025 23:44:11 +0100},
  biburl       = {https://dblp.org/rec/conf/esa/AbrahamsenB0K024.bib},
  bibsource    = {dblp computer science bibliography, https://dblp.org}
}

@inproceedings{DBLP:conf/esa/Bertram25,
  author       = {Christian Bertram},
  _editor       = {Anne Benoit and Haim Kaplan and Sebastian Wild and Grzegorz Herman},
  title        = {Online Metric {TSP}},
  booktitle    = {Proc.\ 33rd European Symposium on Algorithms (ESA)},
  series       = {LIPIcs},
  volume       = {351},
  pages        = {80:1--80:9},
  publisher    = {Schloss Dagstuhl},
  year         = {2025},
  _url          = {https://doi.org/10.4230/LIPIcs.ESA.2025.80},
  doi          = {10.4230/LIPICS.ESA.2025.80},
  timestamp    = {Fri, 21 Nov 2025 23:44:11 +0100},
  biburl       = {https://dblp.org/rec/conf/esa/Bertram25.bib},
  bibsource    = {dblp computer science bibliography, https://dblp.org}
}

@inproceedings{DBLP:conf/soda/000126b,
  author       = {Yang Hu},
  _editor       = {Kasper Green Larsen and Barna Saha},
  title        = {Nearly Optimal Bounds for Stochastic Online Sorting},
  booktitle    = {Proc.\ 37th {ACM-SIAM} Symposium on Discrete Algorithms  (SODA)},
  pages        = {4969--4995},
  _publisher    = {SIAM},
  year         = {2026},
  url          = {https://doi.org/10.1137/1.9781611978971.180},
  doi          = {10.1137/1.9781611978971.180},
  timestamp    = {Thu, 19 Feb 2026 16:57:53 +0100},
  biburl       = {https://dblp.org/rec/conf/soda/000126b.bib},
  bibsource    = {dblp computer science bibliography, https://dblp.org}
}

@inproceedings{DBLP:conf/waoa/NirjhorW25,
  author       = {Jubayer Nirjhor and
                  Nicole Wein},
  _editor       = {Jannik Matuschke and Jos{\'{e}} Verschae},
  title        = {Improved Online Sorting},
  booktitle    = {Proc.\ 23rd Workshop on Approximation and Online Algorithms (WAOA)},
  series       = {LNCS},
  volume       = {16077},
  pages        = {187--197},
  publisher    = {Springer},
  year         = {2025},
  _url          = {https://doi.org/10.1007/978-3-032-06706-7\_13},
  doi          = {10.1007/978-3-032-06706-7_13},
  timestamp    = {Sun, 09 Nov 2025 16:31:20 +0100},
  biburl       = {https://dblp.org/rec/conf/waoa/NirjhorW25.bib},
  bibsource    = {dblp computer science bibliography, https://dblp.org}
}

@inproceedings{DBLP:conf/stacs/KalavasPT26,
  author       = {Andreas Kalavas and
                  Charalampos Platanos and
                  Thanos Tolias},
  _editor       = {Meena Mahajan and Florin Manea and Annabelle McIver and Kim Thang Nguyen},
  title        = {A Polylogarithmic Competitive Algorithm for Stochastic Online Sorting
                  and {TSP}},
  booktitle    = {Proc.\ 43rd Symposium on Theoretical Aspects of Computer Science (STACS)},
  series       = {LIPIcs},
  volume       = {364},
  pages        = {58:1--58:17},
  publisher    = {Schloss Dagstuhl},
  year         = {2026},
  _url          = {https://doi.org/10.4230/LIPIcs.STACS.2026.58},
  doi          = {10.4230/LIPICS.STACS.2026.58},
  timestamp    = {Wed, 25 Feb 2026 17:18:42 +0100},
  biburl       = {https://dblp.org/rec/conf/stacs/KalavasPT26.bib},
  bibsource    = {dblp computer science bibliography, https://dblp.org}
}

@article{diedrich2008approximation,
  title={Approximation algorithms for {3D} orthogonal knapsack},
  author={Diedrich, Florian and Harren, Rolf and Jansen, Klaus and Th{\"o}le, Ralf and Thomas, Henning},
  journal={Journal of Computer Science and Technology},
  volume={23},
  number={5},
  pages={749--762},
  year={2008},
  doi = {10.1007/S11390-008-9170-7}
}

@inproceedings{jansen2014new,
  title={A new asymptotic approximation algorithm for 3-dimensional strip packing},
  author={Jansen, Klaus and Pr{\"a}del, Lars},
  booktitle={Proc.\ 40th International Conference on Current Trends in Theory and Practice of Informatics (SOFSEM)},
  series       = {LNCS},
  volume       = {8327},
  pages        = {327--338},
  publisher    = {Springer},
  year         = {2014},
  doi          = {10.1007/978-3-319-04298-5_29},
}

@article{DBLP:journals/dcg/LassakZ91,
  author       = {Marek Lassak and
                  Jixian Zhang},
  title        = {An On-line Potato-Sack Theorem},
  journal      = {Discret. Comput. Geom.},
  volume       = {6},
  pages        = {1--7},
  year         = {1991},
  url          = {https://doi.org/10.1007/BF02574670},
  doi          = {10.1007/BF02574670},
  bibsource    = {dblp computer science bibliography, https://dblp.org}
}

@article{ausiello2001algorithms,
  author       = {Giorgio Ausiello and
                  Esteban Feuerstein and
                  Stefano Leonardi and
                  Leen Stougie and
                  Maurizio Talamo},
  title        = {Algorithms for the On-Line Travelling Salesman},
  journal      = {Algorithmica},
  volume       = {29},
  number       = {4},
  pages        = {560--581},
  year         = {2001},
  doi          = {10.1007/S004530010071}
}

@article{coppersmith1989multidimensional,
  title={Multidimensional on-line bin packing: algorithms and worst-case analysis},
  author={Coppersmith, Don and Raghavan, Prabhakar},
  journal={Operations Research Letters},
  volume={8},
  number={1},
  pages={17--20},
  year={1989},
  doi={10.1016/0167-6377(89)90027-8}
}

@article{csirik1993line,
  title={An on-line algorithm for multidimensional bin packing},
  author={Csirik, J{\'a}nos and Van Vliet, Andr{\'e}},
  journal={Operations Research Letters},
  volume={13},
  number={3},
  pages={149--158},
  year={1993},
  doi={10.1016/0167-6377(93)90004-Z}
}

@article{DBLP:journals/siamcomp/EpsteinS05,
  author       = {Leah Epstein and
                  Rob van Stee},
  title        = {Optimal Online Algorithms for Multidimensional Packing Problems},
  journal      = {{SIAM} J. Comput.},
  volume       = {35},
  number       = {2},
  pages        = {431--448},
  year         = {2005},
  url          = {https://doi.org/10.1137/S0097539705446895},
  doi          = {10.1137/S0097539705446895},
  bibsource    = {dblp computer science bibliography, https://dblp.org}
}

@article{DBLP:journals/algorithmica/SeidenS03,
  author       = {Steven S. Seiden and
                  Rob van Stee},
  title        = {New Bounds for Multidimensional Packing},
  journal      = {Algorithmica},
  volume       = {36},
  number       = {3},
  pages        = {261--293},
  year         = {2003},
  _url          = {https://doi.org/10.1007/s00453-003-1016-7},
  doi          = {10.1007/S00453-003-1016-7},
  timestamp    = {Wed, 17 May 2017 14:25:13 +0200},
  biburl       = {https://dblp.org/rec/journals/algorithmica/SeidenS03.bib},
  bibsource    = {dblp computer science bibliography, https://dblp.org}
}

@article{DBLP:journals/talg/HanCTZZ11,
  author       = {Xin Han and
                  Francis Y. L. Chin and
                  Hing{-}Fung Ting and
                  Guochuan Zhang and
                  Yong Zhang},
  title        = {A new upper bound 2.5545 on 2D Online Bin Packing},
  journal      = {{ACM} Trans. Algorithms},
  volume       = {7},
  number       = {4},
  pages        = {50:1--50:18},
  year         = {2011},
  url          = {https://doi.org/10.1145/2000807.2000818},
  doi          = {10.1145/2000807.2000818},
  bibsource    = {dblp computer science bibliography, https://dblp.org}
}

@article{DBLP:journals/iandc/HanIYZ16,
  author       = {Xin Han and
                  Kazuo Iwama and
                  Deshi Ye and
                  Guochuan Zhang},
  title        = {Approximate strip packing: {R}evisited},
  journal      = {Inf. Comput.},
  volume       = {249},
  pages        = {110--120},
  year         = {2016},
  _url          = {https://doi.org/10.1016/j.ic.2016.03.010},
  doi          = {10.1016/J.IC.2016.03.010},
  timestamp    = {Tue, 21 Mar 2023 21:09:35 +0100},
  biburl       = {https://dblp.org/rec/journals/iandc/HanIYZ16.bib},
  bibsource    = {dblp computer science bibliography, https://dblp.org}
}

@article{DBLP:journals/acta/BrownBK82,
  author       = {Donna J. Brown and
                  Brenda S. Baker and
                  Howard P. Katseff},
  title        = {Lower Bounds for On-Line Two-Dimensional Packing Algorithms},
  journal      = {Acta Informatica},
  volume       = {18},
  pages        = {207--225},
  year         = {1982},
  url          = {https://doi.org/10.1007/BF00264439},
  doi          = {10.1007/BF00264439},
  bibsource    = {dblp computer science bibliography, https://dblp.org}
}

@article{KernP13,
  author       = {Walter Kern and
                  Jacob Jan Paulus},
  title        = {A tight analysis of {B}rown-{B}aker-{K}atseff sequences for online strip
                  packing},
  journal      = {J. Comb. Optim.},
  volume       = {26},
  number       = {2},
  pages        = {333--344},
  year         = {2013},
  doi          = {10.1007/S10878-012-9463-1}
}

@article{HarrenK15,
  author       = {Rolf Harren and
                  Walter Kern},
  title        = {Improved Lower Bound for Online Strip Packing},
  journal      = {Theory Comput. Syst.},
  volume       = {56},
  number       = {1},
  pages        = {41--72},
  year         = {2015},
  doi          = {10.1007/S00224-013-9494-8}
}

@article{DBLP:journals/tcs/HurinkP11,
  author       = {Johann L. Hurink and
                  Jacob Jan Paulus},
  title        = {Improved online algorithms for parallel job scheduling and strip packing},
  journal      = {Theor. Comput. Sci.},
  volume       = {412},
  number       = {7},
  pages        = {583--593},
  year         = {2011},
  _url          = {https://doi.org/10.1016/j.tcs.2009.05.033},
  doi          = {10.1016/J.TCS.2009.05.033},
  timestamp    = {Wed, 17 Feb 2021 22:00:20 +0100},
  biburl       = {https://dblp.org/rec/journals/tcs/HurinkP11.bib},
  bibsource    = {dblp computer science bibliography, https://dblp.org}
}

\end{document}